\documentclass[11pt]{article}
\usepackage[letterpaper,margin=1in]{geometry}
\usepackage{pgfplotstable}
\usepackage{pgfplots}
\usepackage{xcolor}
\usepackage[colorlinks]{hyperref}
\definecolor{lightblue}{rgb}{0.5,0.5,1.0}
\definecolor{darkred}{rgb}{0.5,0,0}
\definecolor{darkgreen}{rgb}{0,0.5,0}
\definecolor{darkblue}{rgb}{0,0,0.5}

\hypersetup{colorlinks,linkcolor=darkred,filecolor=darkgreen,urlcolor=darkred,citecolor=darkblue}

\usepackage{tabularx}
\usepackage{mathtools}
\usepackage{verbatim}
\usepackage{subcaption}
\usepackage{enumitem}
\usepackage{amsthm}
\usepackage{complexity}
\usepackage{amssymb}
\usepackage{amsmath}
\usepackage{mathtools}
\usepackage[linesnumbered,ruled,commentsnumbered,longend]{algorithm2e}
\usetikzlibrary{shapes.misc, positioning}
\usetikzlibrary{shapes,fit}
\usepackage[edges]{forest}
\usetikzlibrary{shapes,arrows,positioning,calc,decorations.pathmorphing}
\usepackage{amsmath}

\newcommand{\nauty}{\textsc{nauty}}

\DeclareMathOperator{\AND}{AND}
\DeclareMathOperator{\cost}{cost}
\DeclareMathOperator{\opt}{opt}



\newtheorem{theorem}{Theorem}
\newtheorem{lemma}{Lemma}
\newtheorem{corollary}[lemma]{Corollary}

\newtheorem*{problem*}{Problem}

\pgfdeclarelayer{bg} 
\pgfsetlayers{bg,main} 

\title{Comparative Design-Choice Analysis of Color Refinement Algorithms Beyond the Worst Case}

\newcommand\blfootnote[1]{%
  \begingroup
  \renewcommand\thefootnote{}\footnote{#1}%
  \addtocounter{footnote}{-1}%
  \endgroup
}

\begin{document}
\author{Markus Anders \and Pascal Schweitzer \and Florian Wetzels}
\maketitle
\begin{abstract}
	Color refinement is a crucial subroutine in symmetry detection in theory as well as practice. It has further applications in machine learning and in computational problems from linear algebra.
	
	While tight lower bounds for the worst case complexity are known [Berkholz, Bonsma, Grohe, ESA2013] no comparative analysis of design choices for color refinement algorithms is available.
	
	 We devise two models within which we can compare color refinement algorithms using formal methods, an online model and an approximation model. We use these to show that no online algorithm is competitive beyond a logarithmic factor and no algorithm can approximate the optimal color refinement splitting scheme beyond a logarithmic factor.
	 
	 We also directly compare strategies used in practice showing that, on some graphs, queue based strategies outperform stack based ones by a logarithmic factor and vice versa. Similar results hold for strategies based on priority queues.
\end{abstract}\blfootnote{The research leading to these results has received funding from the European Research Council (ERC) under the European Union's Horizon 2020 research and innovation programme (EngageS: grant agreement No.~{820148}).}
\thispagestyle{empty}
\newpage
\setcounter{page}{1}
\section{Introduction}
Color refinement, also known as 1-dimensional Weisfeiler-Leman algorithm, is a crucial cornerstone of symmetry detection in theory as well as practice. 
It emerged as a subroutine for algorithms solving the graph isomorphism problem and its efficiency remains to date one of the determining factors for the running time of practical isomorphism solvers. 
Modern, highly efficient implementations are based on Hopcroft's algorithm for automata minimization~\cite{hop71}, which was first adapted to color refinement by McKay in his widely used tool \nauty~\cite{practicaliso1}.
A more recent but also in the meantime large application area of color refinement can be found in machine learning. Specifically, color refinement is used in the Weisfeiler-Leman Kernel for graph classifications as a measure for similarity~\cite{DBLP:journals/jmlr/ShervashidzeSLMB11} and as the foundation of graph neural networks~\cite{DBLP:conf/aaai/0001RFHLRG19}.
The algorithm can also be applied to effectively reduce the size of linear equation systems~\cite{DBLP:conf/esa/GroheKMS14}.

Given a graph, color refinement iteratively recolors the vertices producing increasingly fine partitions of vertices into color classes. 
Starting with an initial, usually monochromatic coloring,
in each iteration the colors of the vertices are chosen to depend on the colors of the neighbors and their multiplicities.
If vertices differ in the number of neighbors they have in some color class, the algorithm \emph{splits} up the vertices accordingly by assigning them distinct colors.
This is done exhaustively until no further splits are possible.   

The applications mentioned above depend on highly engineered implementations of the algorithm.
This is the reason why modern implementations meticulously optimize the color refinement subroutine treating many special cases with tailored code~\cite{DBLP:conf/alenex/JunttilaK07,practicaliso2, doi:10.1137/1.9781611976472.6}.
Especially in machine learning applications 
it is crucial to achieve scalability for big data inputs~\cite{DBLP:journals/jmlr/ShervashidzeSLMB11}.
Overall, demand for fast implementations of color refinement is high.
Since color refinement has a quasilinear worst case running time, even small logarithmic or constant factors can have a crucial impact.

Indeed, the best known implementation of color refinement runs in time $\mathcal{O}(m \log(n))$ (see~\cite{tightLowerBound,DBLP:conf/icalp/KieferM20}). 
Remarkably, within a model with modest assumptions, a tight lower bound construction matching this upper bound was given in 2015~\cite{tightLowerBound}.
This result tells us that there are graphs for which color refinement, no matter how it is implemented, runs in $\Omega(m \log(n))$. However, the result does not make any comparative statements between various ways to implement color refinement. In fact, there are dramatic differences in the various implementations of color refinement. While all color refinement algorithms depend on performing the aforementioned splits, there is a lot of freedom as to which order we perform the splits in. A \emph{worklist} is usually employed to determine in what order these splits are performed. 
Common choices include a stack, queue, priority queue or combinations of these.

So far however, there has been no rigorous analysis as to whether one worklist choice is superior over another -- or how significant the order of splits actually is.
Going one step further, a natural question is whether there are efficient optimal solutions. If not the case, maybe there are at least solutions that are competitive with all other methods.

\textbf{Contribution.} This paper performs an in-depth comparative analysis of design choices for color refinement algorithms. The first challenge is to actually find a model within which we can compare color refinement algorithms with formal methods.
We employ a two-pronged approach.
We distinguish (1) algorithms that may only use information realistically collected during the color refinement process itself, and (2) algorithms that are allowed to compute additional information about the underlying graph.
Remarkably, our results in the two orthogonal models concur in their conclusion. Namely, that there is no design choice that is competitive beyond a logarithmic factor.

More specifically, in (1) we model algorithms that may only access information explored during the color refinement process itself. 
For this we define a formal online model within which, in fact, all practical algorithms operate.
In this model, the algorithmic decisions of when to refine with respect to what may solely depend on this information.
We prove that this information does not suffice to make optimal or even competitive choices, no matter the amount of computational power used. Specifically, we show no online algorithm is within a logarithmic factor of the offline optimum. We also investigate the direct relationship between practical (online) color refinement strategies. Each of strategies stack, queue, and priority queue, is outperformed by another of the strategies by a logarithmic factor on some graphs.

For (2), we define an ``offline'' version of the problem, which is essentially to compute an optimal split order for a given graph. 
Through a reduction from the set cover problem we prove an approximation hardness result.
Specifically, unless $\P = \NP$, no approximation factor in~$o(\log(n))$ can be achieved by polynomial-time algorithms.
This proves that unless $\P = \NP$, even when collecting more information about the underlying graph than current algorithms actually do, computing a competitive let alone optimal order of splits is intractable.

Overall, our results demonstrate that while the choice of worklist can indeed make a crucial difference, there is no clear optimal color refinement strategy. We conclude that users need to adapt color refinement algorithms to the specific type of graphs encountered in the algorithmic application area in mind.

\section{Color Refinement}
All graphs in this paper are simple, undirected graphs, unless stated otherwise.
The neighborhood of a vertex~$v$ is denoted~$N(v)$. For a set of vertices~$V'\subseteq V(G)$ the \emph{neighborhood} is the set~$N[V']\coloneqq (\cup_{v\in V'} N(v))\setminus V'$. A \emph{coloring} of a graph~$G$ is a map~$\pi\colon V(G)\rightarrow \mathcal{C}$ from the vertices to some set of colors. A (color) \emph{class} is a set~$\pi^{-1}(c)$ of vertices of the same color. 

We begin with a discussion of the color refinement algorithm itself.
Algorithm~\ref{alg:refine} describes a typical rendition of color refinement. The basic idea is as follows. If two vertices in some class~$X$ have a different number of neighbors in some class~$C$ then~$X$ can be split by partitioning it according to neighbor counts in~$C$. 
Whenever we split up a class $X$ according to its connections to another class $C$ in such a fashion (see Line~\ref{line:refine:wrt} and Line~\ref{line:refine:split}) we say that we \emph{refine $X$ with respect to $C$}. 
Specifically this means that after the split, two vertices have the same color precisely if they had the same color before the split and they have the same number of neighbors in~$X$.
We repeatedly split classes with respect to other classes until no further splits are possible. A partition not admitting further splits is called \emph{equitable}. 

\SetKwProg{Fn}{function}{}{end}
\SetKwFunction{Refine}{ColorRefinement}
\begin{algorithm}[t] 
	\SetAlgoLined
	\SetAlgoNoEnd
	\caption[Refinement Procedure]{A typical rendition of color refinement.}\label{alg:refine}
	\Fn{\Refine{G, $\pi$}}{
		\SetKwInOut{Input}{Input}
		\SetKwInOut{Output}{Output}
		\Input{graph $G$, coloring $\pi$}
		\Output{refined coloring $\pi$}
		initialize empty worklist $W$\;
		put all cells of $\pi$ into $W$\;
		\While{{\normalfont $W$ is non-empty}}{
			take a cell $C$ from $W$\;\label{line:extract:cell}
			\For{{\normalfont each cell $X$ containing a neighbor of a vertex in $C$}}{
				for each vertex in $X$ count its neighbors in $C$ \; \label{line:refine:wrt}
				split $X$ into $X_1, \dots{}, X_k$ in $\pi$, according to neighbor counts\; \label{line:refine:split}
				let $X_i$ be one of the largest cells of $X_1, \dots{}, X_k$\; \label{line:refine:largest}
				put all sets $X_1, \dots{}, X_k$ except $X_i$ into $W$\; \label{line:refine:putstack}
				\lIf{$X \in W$}{replace $X$ in $W$ with $X_i$}
			}
		}
		\Return{$\pi$}
	}
\end{algorithm}

Algorithm~\ref{alg:refine} maintains the classes with respect to which refinements still have to be performed in a worklist~$W$.
Note that the algorithm does not fully specify the internals of the worklist. Specifically, it does not state in Line~\ref{line:extract:cell} which cell is extracted from the worklist. We should emphasize that the final partition into color classes is independent of the choices of cells that are extracted, however the overall running time may depend on it.
Typical implementations use a stack, queue, priority queue or a similar data structure. All of these choices result in the same 
worst case running time of~$\Theta((n+m) (\log n))$ (see~\cite{tightLowerBound}). To achieve this running time it is crucial to prevent one largest cell (Line~\ref{line:refine:largest}) from being added to the worklist. Splits with respect to this class are already covered by the other classes.

Overall, the main design choice of the algorithm is the choice of when to split which class with which other class. To describe a general framework for the possible strategies of what to split when, we first need to understand what information is available to the algorithm for making its decision.

\subsection{Partial Quotient Graphs}
For an equitable partition, quotient graphs capture the information of how many neighbors vertices from one class have in another class. They are used in so-called individualization-refinement algorithms as pruning invariants (see~\cite{practicaliso2}). Typically, the quotient graph is computed on the fly during the execution of a color refinement algorithm. 

We now introduce the concept of \emph{partial quotient graphs}. 
These graphs are a tool to formalize the information gathered up to a certain point during the execution of color refinement algorithms. 
As we cannot precisely say which information an algorithm collects, 
the quotient graphs give an overapproximation of the available information and model all information that could have possibly been gathered. For the purpose of our lower bounds, overapproximating can only strengthen the conclusions.

The partial quotient graph of a colored graph $(G, \pi)$ is denoted by $P(G, \pi)$. 
Quotient graphs are directed and contain self-loops.
They include vertex labels $l_V$ as well as edge labels $l_E$. The vertex set of $P(G, \pi)$ is the set of all sets of colors of~$(G,\pi)$, i.e., $V(P(G, \pi)) := 2^{\pi(V(G))}$. 
A set of colors represents the class that is the union of the respective color classes.

Vertices of the partial quotient graph are labeled with the size of their corresponding set of vertices in $G$, i.e., for all sets of colors $c \in 2^{\pi(V(G))}$ we define $l_V(P(G, \pi))(c) := |\pi^{-1}(c)|$, where by~$\pi^{-1}(c)$ we denote the vertices whose color is in~$c$.
The edge set contains all connections between (unions of) color classes that would not cause a split.  
Thus there is an edge from~$c_1$ to~$c_2$ if~$\pi^{-1}(c_2)$ does not split~$\pi^{-1}(c_1)$.
Formally, this means
\begin{align*}
	E(P(G, \pi)) := \{(c_1, c_2) \;|\;& c_1, c_2 \in 2^{\pi(V(G))},\forall v,w \in \pi^{-1}(c_1):\; d_{\pi^{-1}(c_2)}(v) = d_{\pi^{-1}(c_2)}(w)\}.
\end{align*}
Edges only exist whenever the connection between unions of color classes are regular on one side, so we can label each edge with the corresponding degree, i.e., 
$l_E(P(G, \pi))((c_1, c_2)) := d_{\pi^{-1}(c_2)}(v)$, where $v \in \pi^{-1}(c_1)$ is arbitrary.

Let us justify the definition with an example. Suppose we split in a monochromatic graph the class of all vertices with itself. Then the new coloring partitions the vertices precisely by degree. That is, classes contain vertices of the same degree. An algorithm would know this degree, since it has counted the edges incident with each vertex, but it would not know how many neighbors a vertex has within a current color class. In the partial quotient graph, there is an edge from each new color classes to the union of all color classes.

The definition of partial quotient graphs contains many more vertices and edges and information on these than would truly be available while executing color refinement.
In fact, partial quotient graphs grow exponentially in size, since all possible unions of color classes are considered.
Common color refinement algorithms clearly gather much less information. Firstly, only connections of classes that are involved in a refinement are actually considered.
Secondly, only information about unions of colors that occurred as a color class in a previous step of the refinement is known.
Thus, usually color refinement algorithms only uncover a small, polynomial-sized portion of the partial quotient graphs defined above.

However, for our lower bounds, we assume that algorithms have access to the entire partial quotient graphs. 
We show that even if we generously allow such access, the information is not sufficient to derive a strategy with constant competitive ratio. 
For upper bounds, we only use information of the aforementioned polynomial-sized portion of partial quotient graphs. In fact, the upper bounds are based on a stack-based approach akin to Algorithm~\ref{alg:refine}. 
\subsection{Online Model}
We now define a model that bases the choice of which color classes to use for the next refinement solely on the information available through partial quotient graphs.  
Practical implementations such as a queue or a stack are naturally captured by this, but the model even allows for much more powerful choices. 
The goal is then to prove that no strategy based solely on information of partial quotient graphs is sufficient to make optimal choices.

Let us start by defining the concept of a \emph{strategy} $W: \mathcal{P}^* \to (2^\mathbb{N})^2$. A strategy is a function mapping a string of quotient graphs $P =P_1 \cdots{} P_k\in \mathcal{P}^*$ to two vertices of the last quotient graph $(C,X) \in V(P_k)^2$, that is, two unions of color classes. 
The string of graphs $P$ denotes all partial quotient graphs observed during execution of the algorithm up to step $k$. 
The pair $(C, X)$ denotes the choice of colors with which the algorithm continues in the next step: in step $k + 1$, the algorithm refines $X$ with respect to $C$. 

\begin{algorithm}[t] 
	\SetAlgoLined
	\SetAlgoNoEnd
	\caption[Refinement Procedure]{Corresponding color refinement for a strategy $W$.}\label{alg:refinecorr}
	\Fn{\Refine{G, $\pi$}}{
		\SetKwInOut{Input}{Input}
		\SetKwInOut{Output}{Output}
		\Input{graph $G$, coloring $\pi$}
		\Output{refined coloring $\pi$}
		create list $S$ containing~$P(G, \pi)$\;
		\While{{$\pi$ is not equitable}}{
			$(C, X) := W(S)$\;
			for each vertex in $X$ count its neighbors in $C$\;
			split $X$ into $X_1, \dots{}, X_k$ in $\pi$, according to neighbor counts\; \label{line:refinecorr:split}
			append $P(G, \pi)$ to~$S$\;	
		}
		\Return{$\pi$}
	}
\end{algorithm}

For a strategy $W$ we now define a \emph{corresponding color refinement implementation}. 
Assume we are working on $G$ and have already refined up to a coloring $\pi_k$ within $k$ steps.
Furthermore, let $P_1, \dots{}, P_k$ denote the partial quotient graphs corresponding to the execution.
Next, we compute $(C, X) = W(P_1\cdots{}P_k)$ and refine $X$ with respect to $C$.
The algorithm terminates whenever $\pi_k$ is equitable.
A formal definition is given in Algorithm~\ref{alg:refinecorr}.
We call $W$ a \emph{valid} strategy if the corresponding color refinement implementation is correct, i.e., if it terminates with an equitable partition in finite time on all finite graphs.

Throughout this paper, we measure the \emph{cost} of the strategy~$W$, denoted~$\cost(W, G)$, in terms of the number of edges that need to be considered to execute the refinements. 
Specifically, when refining $X$ with respect to $C$, we charge the algorithm the number of edges connecting $X$ with $C$. This is the same model as used in \cite{tightLowerBound} reflecting the actual running time of practical implementations (see~\cite{practicaliso1,practicaliso2}). We use the terms cost and time interchangeably.

\section{Graph Gadgets}\label{sec:gadgets}
Throughout the paper we construct graphs that cause color refinement to behave in particular manners. 
These graphs are mostly built using three types of graph gadgets, described next.

\textbf{And gadgets.}
Let us first discuss the $\AND_i$ gadgets as used by Berkholz et al.~\cite{tightLowerBound}. There is set~$B$ of $2^{i}$ in-vertices that come in pairs and $2$ out-vertices.
The goal of the gadget is that whenever all pairs of \emph{in-vertices} have been split, a split of two \emph{out-vertices} $a_0$ and $a_1$ is induced, but not before. 

The $\AND_2$ gadget (see Figure~\ref{fig:basicGadgets}) is the well known CFI-gadget~\cite{CFI}, where two gates form the in-vertices $B$ and the third one the out-vertices $a_0, a_1$.

The $\AND_i$ gadget is constructed recursively using $\AND_2$ gadgets. 
For $i>2$, the $\AND_i$ gadget is constructed by taking the union of two $\AND_{i-1}$ and one $\AND_2$ gadget. The four out-vertices of the $\AND_{i-1}$ gadgets are then connected to the four in-vertices of the $\AND_2$ gadget. 
Figure~\ref{fig:basicGadgets} shows how the $\AND_3$ gadget can be constructed using three $\AND_2$ gadgets.

The important property is that in an $\AND_i$ gadget, all pairs $b_{2j},b_{2j+1}$ with $j \in \{0,...,2^{i-1}\}$ need to be distinguished to induce a split of $a_0$ and $a_1$.  We should also record a property for the opposite direction: if $a_0$ and $a_1$ are distinguished, no split on $B$ should be induced.

\textbf{Unidirectional gadgets.}
We now describe the \emph{undirectional gadget}. 
As the name suggests, it blocks the continuation of a split of pairs in one direction but allows it in the opposite direction. 
Figure~\ref{fig:basicGadgets} illustrates the gadget.

The gadget behaves as follows. 
Consider in-vertices $b_0$, $b_1$ and out-vertices $a_0$ and $a_1$. Distinguishing $b_0$ and $b_1$ should induce a split of $a_0$ and $a_1$. 
However, distinguishing $a_0$ and $a_1$ should \emph{not} cause a split of $b_0$ and $b_1$.
The gadget is obtained through a modification of the $\AND_2$ gadget. 
We use the fact that a split of out-vertices in $\AND_2$ does not cause a split of the pairs of in-vertices. 
Therefore, by connecting the in-vertices to new vertices $a_0$ and $a_1$, such that the $\AND_2$ gadget is activated by any of the two singletons, we get the desired property. 

Interestingly, the unidirectional gadget has also been used as a crucial building block in~\cite{DBLP:conf/mfcs/ArvindFKKR16} and~\cite{DBLP:conf/focs/Grohe96} to study the complexity of various problems closely related to color refinement.

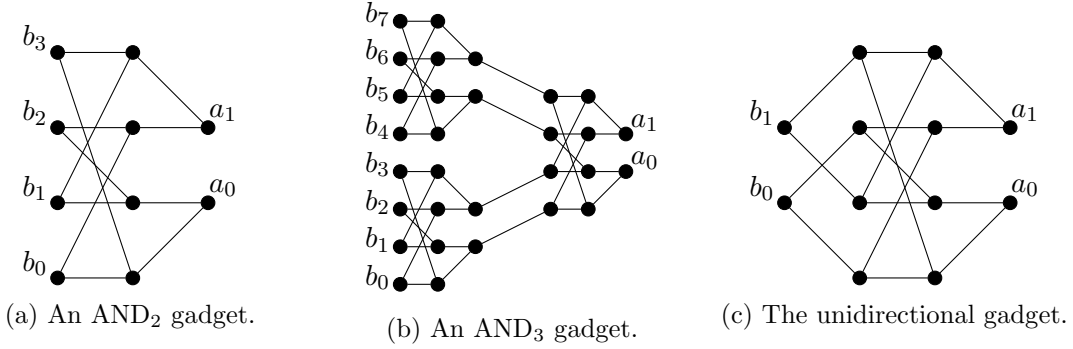
\begin{figure}[]
	
	\centering
	\begin{subfigure}{0.3\linewidth}
		\centering
		\begin{tikzpicture}
		\node[draw,circle,fill,scale=0.5] at (0,0) (b_0) {};
		\node[draw,circle,fill,scale=0.5] at (0,1) (b_1) {};
		\node[draw,circle,fill,scale=0.5] at (0,2) (b_2) {};
		\node[draw,circle,fill,scale=0.5] at (0,3) (b_3) {};
		
		\node[draw,circle,fill,scale=0.5] at (1,0) (c_0) {};
		\node[draw,circle,fill,scale=0.5] at (1,1) (c_1) {};
		\node[draw,circle,fill,scale=0.5] at (1,2) (c_2) {};
		\node[draw,circle,fill,scale=0.5] at (1,3) (c_3) {};
		
		\node[draw,circle,fill,scale=0.5] at (2,1) (a_0) {};
		\node[draw,circle,fill,scale=0.5] at (2,2) (a_1) {};
		
		\node[] at (-0.3,0.2) (b_0_label) {$b_0$};
		\node[] at (-0.3,1.2) (b_1_label) {$b_1$};
		\node[] at (-0.3,2.2) (b_2_label) {$b_2$};
		\node[] at (-0.3,3.2) (b_3_label) {$b_3$};
				
		\node[] at (2.2,1.2) (a_0_label) {$a_0$};
		\node[] at (2.2,2.2) (a_1_label) {$a_1$};
		
		\node[fill=none,stroke=none] at (0,3.5) (dummy) {};
		
		\draw (b_0) -- (c_0);
		\draw (b_0) -- (c_2);
		\draw (b_1) -- (c_1);
		\draw (b_1) -- (c_3);
		\draw (b_2) -- (c_1);
		\draw (b_2) -- (c_2);
		\draw (b_3) -- (c_0);
		\draw (b_3) -- (c_3);
		
		\draw (c_0) -- (a_0);
		\draw (c_1) -- (a_0);
		\draw (c_2) -- (a_1);
		\draw (c_3) -- (a_1);
		\end{tikzpicture}
		\caption{An $\AND_2$ gadget.}
		\label{fig:AND2}
	\end{subfigure}
	\begin{subfigure}{0.3\linewidth}
		\centering
		\begin{tikzpicture}[scale=0.5]
		\node[draw,circle,fill,scale=0.5] at (0,0) (b_0_0) {};
		\node[draw,circle,fill,scale=0.5] at (0,1) (b_1_0) {};
		\node[draw,circle,fill,scale=0.5] at (0,2) (b_2_0) {};
		\node[draw,circle,fill,scale=0.5] at (0,3) (b_3_0) {};
		
		\node[draw,circle,fill,scale=0.5] at (1,0) (c_0_0) {};
		\node[draw,circle,fill,scale=0.5] at (1,1) (c_1_0) {};
		\node[draw,circle,fill,scale=0.5] at (1,2) (c_2_0) {};
		\node[draw,circle,fill,scale=0.5] at (1,3) (c_3_0) {};
		
		\node[draw,circle,fill,scale=0.5] at (2,1) (a_0_0) {};
		\node[draw,circle,fill,scale=0.5] at (2,2) (a_1_0) {};
		
		\node[draw,circle,fill,scale=0.5] at (0,4) (b_4_0) {};
		\node[draw,circle,fill,scale=0.5] at (0,5) (b_5_0) {};
		\node[draw,circle,fill,scale=0.5] at (0,6) (b_6_0) {};
		\node[draw,circle,fill,scale=0.5] at (0,7) (b_7_0) {};
		
		\node[draw,circle,fill,scale=0.5] at (1,4) (c_4_0) {};
		\node[draw,circle,fill,scale=0.5] at (1,5) (c_5_0) {};
		\node[draw,circle,fill,scale=0.5] at (1,6) (c_6_0) {};
		\node[draw,circle,fill,scale=0.5] at (1,7) (c_7_0) {};
		
		\node[draw,circle,fill,scale=0.5] at (2,5) (a_2_0) {};
		\node[draw,circle,fill,scale=0.5] at (2,6) (a_3_0) {};

		\node[draw,circle,fill,scale=0.5] at (4,2) (b_0_1) {};
		\node[draw,circle,fill,scale=0.5] at (4,3) (b_1_1) {};
		\node[draw,circle,fill,scale=0.5] at (4,4) (b_2_1) {};
		\node[draw,circle,fill,scale=0.5] at (4,5) (b_3_1) {};
		
		\node[draw,circle,fill,scale=0.5] at (5,2) (c_0_1) {};
		\node[draw,circle,fill,scale=0.5] at (5,3) (c_1_1) {};
		\node[draw,circle,fill,scale=0.5] at (5,4) (c_2_1) {};
		\node[draw,circle,fill,scale=0.5] at (5,5) (c_3_1) {};
		
		\node[draw,circle,fill,scale=0.5] at (6,3) (a_0_1) {};
		\node[draw,circle,fill,scale=0.5] at (6,4) (a_1_1) {};
		
		\node[] at (-0.5-0.1,0.2) (b_0_label) {$b_0$};
		\node[] at (-0.5-0.1,1.2) (b_1_label) {$b_1$};
		\node[] at (-0.5-0.1,2.2) (b_2_label) {$b_2$};
		\node[] at (-0.5-0.1,3.2) (b_3_label) {$b_3$};
		\node[] at (-0.5-0.1,4.2) (b_4_label) {$b_4$};
		\node[] at (-0.5-0.1,5.2) (b_5_label) {$b_5$};
		\node[] at (-0.5-0.1,6.2) (b_6_label) {$b_6$};
		\node[] at (-0.5-0.1,7.2) (b_7_label) {$b_7$};
		
		\node[] at (6.3+0.2,3.3) (a_0_label) {$a_0$};
		\node[] at (6.3+0.2,4.3) (a_1_label) {$a_1$};
		
		\draw (b_0_0) -- (c_0_0);
		\draw (b_0_0) -- (c_2_0);
		\draw (b_1_0) -- (c_1_0);
		\draw (b_1_0) -- (c_3_0);
		\draw (b_2_0) -- (c_1_0);
		\draw (b_2_0) -- (c_2_0);
		\draw (b_3_0) -- (c_0_0);
		\draw (b_3_0) -- (c_3_0);
		
		\draw (b_4_0) -- (c_4_0);
		\draw (b_4_0) -- (c_6_0);
		\draw (b_5_0) -- (c_5_0);
		\draw (b_5_0) -- (c_7_0);
		\draw (b_6_0) -- (c_5_0);
		\draw (b_6_0) -- (c_6_0);
		\draw (b_7_0) -- (c_4_0);
		\draw (b_7_0) -- (c_7_0);
		
		\draw (c_0_0) -- (a_0_0);
		\draw (c_1_0) -- (a_0_0);
		\draw (c_2_0) -- (a_1_0);
		\draw (c_3_0) -- (a_1_0);
		
		\draw (c_4_0) -- (a_2_0);
		\draw (c_5_0) -- (a_2_0);
		\draw (c_6_0) -- (a_3_0);
		\draw (c_7_0) -- (a_3_0);

		\draw (b_0_1) -- (c_0_1);
		\draw (b_0_1) -- (c_2_1);
		\draw (b_1_1) -- (c_1_1);
		\draw (b_1_1) -- (c_3_1);
		\draw (b_2_1) -- (c_1_1);
		\draw (b_2_1) -- (c_2_1);
		\draw (b_3_1) -- (c_0_1);
		\draw (b_3_1) -- (c_3_1);
		
		\draw (c_0_1) -- (a_0_1);
		\draw (c_1_1) -- (a_0_1);
		\draw (c_2_1) -- (a_1_1);
		\draw (c_3_1) -- (a_1_1);

		\draw (a_0_0) -- (b_0_1);
		\draw (a_1_0) -- (b_1_1);
		\draw (a_2_0) -- (b_2_1);
		\draw (a_3_0) -- (b_3_1);
		\end{tikzpicture}
		\caption{An $\AND_3$ gadget.}
		\label{fig:AND3}
	\end{subfigure}
	\begin{subfigure}{0.3\linewidth}
		\centering
		\begin{tikzpicture}
		\node[draw,circle,fill,scale=0.5] at (0,0) (b_0) {};
		\node[draw,circle,fill,scale=0.5] at (0,1) (b_1) {};
		\node[draw,circle,fill,scale=0.5] at (0,2) (b_2) {};
		\node[draw,circle,fill,scale=0.5] at (0,3) (b_3) {};
		
		\node[draw,circle,fill,scale=0.5] at (1,0) (c_0) {};
		\node[draw,circle,fill,scale=0.5] at (1,1) (c_1) {};
		\node[draw,circle,fill,scale=0.5] at (1,2) (c_2) {};
		\node[draw,circle,fill,scale=0.5] at (1,3) (c_3) {};
		
		\node[draw,circle,fill,scale=0.5] at (2,1) (a_0) {};
		\node[draw,circle,fill,scale=0.5] at (2,2) (a_1) {};
		
		\node[draw,circle,fill,scale=0.5] at (-1,1) (a_0') {};
		\node[draw,circle,fill,scale=0.5] at (-1,2) (a_1') {};
		
		\node[] at (2.2,1.2) (a_0_label) {$a_0$};
		\node[] at (2.2,2.2) (a_1_label) {$a_1$};
		
		\node[] at (-1.3,1.2) (a_0'_label) {$b_0$};
		\node[] at (-1.3,2.2) (a_1'_label) {$b_1$};
		
		\node[fill=none,stroke=none] at (0,3.5) (dummy) {};
		
		\draw (b_0) -- (c_0);
		\draw (b_0) -- (c_2);
		\draw (b_1) -- (c_1);
		\draw (b_1) -- (c_3);
		\draw (b_2) -- (c_1);
		\draw (b_2) -- (c_2);
		\draw (b_3) -- (c_0);
		\draw (b_3) -- (c_3);
		
		\draw (c_0) -- (a_0);
		\draw (c_1) -- (a_0);
		\draw (c_2) -- (a_1);
		\draw (c_3) -- (a_1);
		
		\draw (b_0) -- (a_0');
		\draw (b_2) -- (a_0');
		\draw (b_1) -- (a_1');
		\draw (b_3) -- (a_1');
		\end{tikzpicture}
		\caption{The unidirectional gadget.}
		\label{fig:unidirectional_gadget}
	\end{subfigure}
	
	\caption{Basic gadget constructions as used throughout the paper. Vertices labeled with $b_i$ always denote in-vertices, while $a_i$ denotes out-vertices.}
	\label{fig:basicGadgets}
\end{figure}

\textbf{Concealer gadgets.}
We conclude our discussion of gadgets with the \emph{concealer gadgets}. 
Similar to the $\AND_i$ gadget, a concealer gadget $C_i$ of level $i$ has $2^i$ in-vertices $B$ and $2$ out-vertices $a_0, a_1$. 
Whereas in the $\AND$ gadget, \emph{all} input pairs need to be distinguished, the concealer gadget only includes \emph{one} specific pair that causes a split of the out-vertices.
We call the pair causing the split of out-vertices the \emph{correct pair}, while all other pairs not causing the split are called \emph{dead end pairs}.

The idea is that the correct pair can not be located easily by color refinement algorithms. 
Hence, the gadget \emph{conceals} where refinement can be continued. 

To achieve this behavior, the gadget consists of $2^{i-1}$ unidirectional gadgets and the out-vertices $a_0, a_1$.
We modify all but one of the unidirectional gadgets so that the connection of the in-gate agrees with the one of the out-gate.
This causes these gadgets to become dead ends -- activating any of these gadgets has no effect on the out-vertices. 
The last, unmodified unidirectional gadget is the only one that can actually split the out-vertices and is therefore the only correct gadget.

The out-vertices of the entire concealer gadget are then connected to the out-vertices of all the unidirectional gadgets so that activating the correct pair causes a split of the out-vertices. 
Figure~\ref{fig:concealer_gadget} shows a concealer gadget $C_3$.

Since we did not specify which of the pairs is the correct pair, there are several concealer gadgets for each $i \in \mathbb{N}$. Abusing notation we denote all of them by $C_i$.
The concealer gadgets have two crucial properties. First, as long as the correct pair has not been split (and the neighbors of a correct pair have not been split) the partial quotient graphs of two concealer gadgets on the same size are isomorphic. Second, the correct pair can only be split from outside the gadget.
We formalize these properties in the following.

Consider two colored concealer gadgets $(C_i, \pi), (C_i', \pi')$ of the same order. Suppose $\{b_s,b_{s+1}\}$ is the correct pair in~$(C_i, \pi)$ and~$\{b_t,b_{t+1}\}$ is the correct pair in~$(C'_i, \pi')$. We say the two graphs still \emph{concur} if the colors for the vertices agree (note that the two graphs have the same vertex set) and in both graphs neither the correct pairs nor their neighbors have been split. Specifically, we require that
	\begin{itemize}
		\item the vertex colorings agree, (i.e.,~$\pi(v) = \pi'(v)$ for every~$v\in V(C_i)=V(C'_i)$),
	
		\item  the correct pairs have not been distinguished (i.e.,~$\pi(b_s)=\pi(b_{s+1})$ and~$\pi'(b_t)=\pi'(b_{t+1})$),
			
		\item the neighbors of the correct pairs have not been distinguished (i.e.,~$\pi(v) = \pi(v')$ for all~$v,v'\in N_{C_i}(b_s)\cup N_{C_i}(b_{s+1})$ and~$\pi(v) = \pi(v')$ for all~$v,v'\in  N_{C'_i}(b_t)\cup N_{C'_i}(b_{t+1})$).
	\end{itemize}

\begin{lemma} \label{lem:concealer} 
	Suppose $(C_i, \pi)$ and~$(C_i', \pi')$ are colored concealer gadgets that concur. Then the graphs have the same partial quotient graphs, i.e., $P(C_i, \pi) = P(C_i', \pi')$.
\end{lemma}
\begin{proof}
	Suppose for a vertex~$v$ we want to count the number of neighbors that~$v$ has in a union of color classes~$X$. We claim that this number is the same in~$C_i$ and~$C'_i$. Indeed, we only need to consider edges incident with~$v$ that have one endpoint in~$M=\{b_s,b_{s+1},b_t,b_{t+1}\}$ and one endpoint in~$N[M]$ (the neighborhood of~$M$). Let~$E'$ be the set of these edges and let~$E'_v$ be the set of these edges incident with~$v$.
	
	Note that for each of the four sets~$\{b_s,b_{s+1}\}$,~$\{b_t,b_{t+1}\}$,~$N[\{b_s,b_{s+1}\}]$, and~$N[\{b_t,b_{t+1}\}]$ either~$X$ contains the set entirely or not at all. 
	
	If~$v$ is in~$M$ then either all edges of~$E'_v$ have an endpoint in~$X$ or no such edge does. 
	
	Likewise if~$v$ is in~$N[M]$ then either all edges of~$E'_v$ have an endpoint in~$X$ or no such edge does.

	Moreover, in either case, whether all such edges are or no such edge is contained does not depend on whether we consider~$C_i$ or~$C'_i$.
	
	This implies that the number of edges counted in the refinement  (i.e., those incident with~$v$ and having an endpoint in~$X$) is the same in~$C_i$ and~$C'_i$.
\end{proof}

\begin{lemma} \label{lem:concealer:refine} 
	For concealer gadgets $(C_i, \pi)$ and~$(C_i', \pi')$ suppose~$\pi=\pi'$ so that
	\begin{itemize}
	\item vertices in an input pair that is correct in one of the graphs have the same color and
	\item all vertices that are not in an input pair have the same color.
	\end{itemize}
Then $(C_i, \pi)$ and~$(C_i', \pi')$ concur. After an arbitrary sequence of splits to both graphs the resulting graphs still concur and neither correct input pairs nor the out pair are split.
\end{lemma}
\begin{proof}
	This follows by induction on the number of steps observing that the functionality of the unidirectional gadget ensures that the output pair is never split, and thus vertices inside correct gadgets are never split.
\end{proof}

The two lemmas show that unless a correct pair is split, the gadgets always concur and an algorithm in the online model will have to perform splits consistently on both graphs. Moreover, the output pair is never split. 

Intuitively this means that in the online model, an algorithm can only guess which pair is the correct pair. 
Therefore, when faced with a concealer gadget, the algorithm potentially has to try all input pairs.

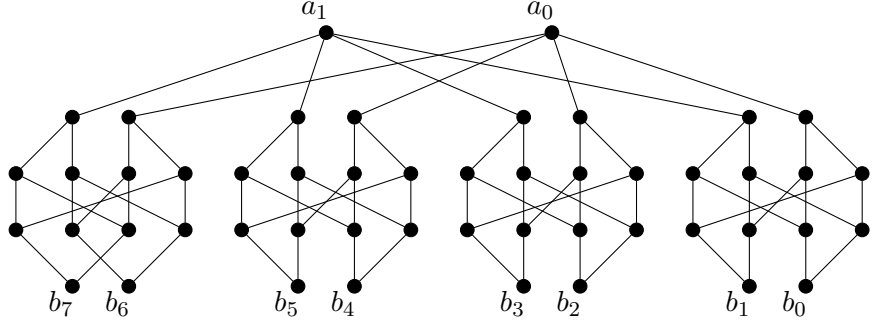
\begin{figure}[]
	\centering
	\begin{tikzpicture}[rotate=90,scale=0.75]
	
	\node[draw,circle,fill,scale=0.5] at (4-0.5,5.5) (a_0) {};
	\node[draw,circle,fill,scale=0.5] at (4-0.5,9.5) (a_1) {};
	
	\foreach \i in {0,...,2}{
	
		\node[draw,circle,fill,scale=0.5] at (0,4*\i+0) (ud\i_b_0) {};
		\node[draw,circle,fill,scale=0.5] at (0,4*\i+1) (ud\i_b_1) {};
		\node[draw,circle,fill,scale=0.5] at (0,4*\i+2) (ud\i_b_2) {};
		\node[draw,circle,fill,scale=0.5] at (0,4*\i+3) (ud\i_b_3) {};
		
		\node[draw,circle,fill,scale=0.5] at (1,4*\i+0) (ud\i_c_0) {};
		\node[draw,circle,fill,scale=0.5] at (1,4*\i+1) (ud\i_c_1) {};
		\node[draw,circle,fill,scale=0.5] at (1,4*\i+2) (ud\i_c_2) {};
		\node[draw,circle,fill,scale=0.5] at (1,4*\i+3) (ud\i_c_3) {};
		
		\node[draw,circle,fill,scale=0.5] at (2,4*\i+1) (ud\i_a_0) {};
		\node[draw,circle,fill,scale=0.5] at (2,4*\i+2) (ud\i_a_1) {};
		
		\node[draw,circle,fill,scale=0.5] at (-1,4*\i+1) (ud\i_a_0') {};
		\node[draw,circle,fill,scale=0.5] at (-1,4*\i+2) (ud\i_a_1') {};
		
		\draw (ud\i_b_0) -- (ud\i_c_0);
		\draw (ud\i_b_0) -- (ud\i_c_2);
		\draw (ud\i_b_1) -- (ud\i_c_1);
		\draw (ud\i_b_1) -- (ud\i_c_3);
		\draw (ud\i_b_2) -- (ud\i_c_1);
		\draw (ud\i_b_2) -- (ud\i_c_2);
		\draw (ud\i_b_3) -- (ud\i_c_0);
		\draw (ud\i_b_3) -- (ud\i_c_3);
		
		\draw (ud\i_c_0) -- (ud\i_a_0);
		\draw (ud\i_c_1) -- (ud\i_a_0);
		\draw (ud\i_c_2) -- (ud\i_a_1);
		\draw (ud\i_c_3) -- (ud\i_a_1);
		
		\draw (ud\i_b_0) -- (ud\i_a_0');
		\draw (ud\i_b_1) -- (ud\i_a_0');
		\draw (ud\i_b_2) -- (ud\i_a_1');
		\draw (ud\i_b_3) -- (ud\i_a_1');
		
		\draw (a_0) -- (ud\i_a_0);
		\draw (a_1) -- (ud\i_a_1);
	
	}

	\node[draw,circle,fill,scale=0.5] at (0,4*3+0) (ud3_b_0) {};
	\node[draw,circle,fill,scale=0.5] at (0,4*3+1) (ud3_b_1) {};
	\node[draw,circle,fill,scale=0.5] at (0,4*3+2) (ud3_b_2) {};
	\node[draw,circle,fill,scale=0.5] at (0,4*3+3) (ud3_b_3) {};
	
	\node[draw,circle,fill,scale=0.5] at (1,4*3+0) (ud3_c_0) {};
	\node[draw,circle,fill,scale=0.5] at (1,4*3+1) (ud3_c_1) {};
	\node[draw,circle,fill,scale=0.5] at (1,4*3+2) (ud3_c_2) {};
	\node[draw,circle,fill,scale=0.5] at (1,4*3+3) (ud3_c_3) {};
	
	\node[draw,circle,fill,scale=0.5] at (2,4*3+1) (ud3_a_0) {};
	\node[draw,circle,fill,scale=0.5] at (2,4*3+2) (ud3_a_1) {};
	
	\node[draw,circle,fill,scale=0.5] at (-1,4*3+1) (ud3_a_0') {};
	\node[draw,circle,fill,scale=0.5] at (-1,4*3+2) (ud3_a_1') {};
	
	\draw (ud3_b_0) -- (ud3_c_0);
	\draw (ud3_b_0) -- (ud3_c_2);
	\draw (ud3_b_1) -- (ud3_c_1);
	\draw (ud3_b_1) -- (ud3_c_3);
	\draw (ud3_b_2) -- (ud3_c_1);
	\draw (ud3_b_2) -- (ud3_c_2);
	\draw (ud3_b_3) -- (ud3_c_0);
	\draw (ud3_b_3) -- (ud3_c_3);
	
	\draw (ud3_c_0) -- (ud3_a_0);
	\draw (ud3_c_1) -- (ud3_a_0);
	\draw (ud3_c_2) -- (ud3_a_1);
	\draw (ud3_c_3) -- (ud3_a_1);
	
	\draw (ud3_b_0) -- (ud3_a_0');
	\draw (ud3_b_2) -- (ud3_a_0');
	\draw (ud3_b_1) -- (ud3_a_1');
	\draw (ud3_b_3) -- (ud3_a_1');
	
	\draw (a_0) -- (ud3_a_0);
	\draw (a_1) -- (ud3_a_1);

	\node[] at (-1.3,1.2) (a_0'_label) {$b_0$};
	\node[] at (-1.3,2.2) (a_1'_label) {$b_1$};
	\node[] at (-1.3,5.2) (a_2'_label) {$b_2$};
	\node[] at (-1.3,6.2) (a_3'_label) {$b_3$};
	\node[] at (-1.3,9.2) (a_4'_label) {$b_4$};
	\node[] at (-1.3,10.2) (a_5'_label) {$b_5$};
	\node[] at (-1.3,13.2) (a_6'_label) {$b_6$};
	\node[] at (-1.3,14.2) (a_7'_label) {$b_7$};
	
	\node[] at (4.3-0.4,5.7) (a_0_label) {$a_0$};
	\node[] at (4.3-0.4,9.7) (a_1_label) {$a_1$};
	
	\end{tikzpicture}
	\caption{A concealer gadget $C_3$. Vertices $b_6, b_7$ form the correct pair; other pairs are dead ends.}
	\label{fig:concealer_gadget}
\end{figure}

\section{Competitive Ratio}\label{sec:competitiveRatio}
We prove the non-existence of a $c$-competitive strategy in the online model.
In particular, in this section, we prove the following theorem:
\begin{theorem}\label{thm:competitive:ratio} \label{thm:competitive} For every strategy $W$ of the online model, there is an infinite family of graphs $G_k$ ($k \in \mathbb{N}$) such that  
	$\cost(W, G_k) \in \Omega(\opt(G_k) \cdot \log(\opt(G_k)))$,
where $\opt(G_k)\in \Theta(|G_k|)$ is the minimal cost of a strategy on $G_k$.
\end{theorem}
The theorem implies that the information provided by partial quotient graphs is not sufficient to make competitive let alone optimal choices in color refinement algorithms.

\newif\ifcol
\newif\ifcoltwo
\coltrue
\coltwotrue
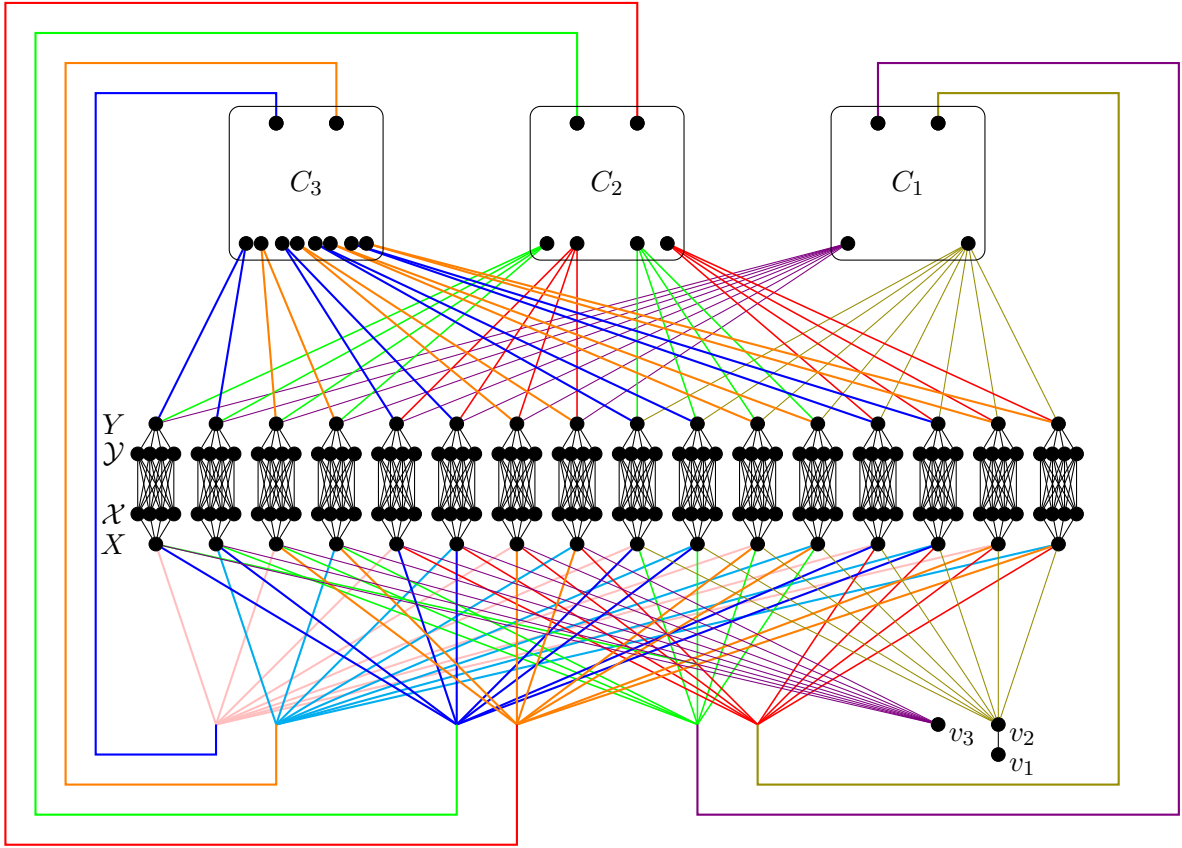
\begin{figure}[t]
	\centering
	\begin{tikzpicture}[rotate=90,scale=0.8]
	\node[draw,circle,fill,scale=0.5] at (-4, 3) (v_3) {};
	\node[] at (-4.2, 2.7-0.1) (v_3_label) {$v_3$};
	\node[draw,circle,fill,scale=0.5] at (-4, 2) (v_2) {};
	\node[] at (-4.2, 1.7-0.1) (v_2_label) {$v_2$};
	\node[draw,circle,fill,scale=0.5] at (-4.5, 2) (v_1) {};
	\node[] at (-4.7, 1.7-0.1) (v_1_label) {$v_1$};
	\draw[] (v_1) -- (v_2);
	
	\node[] at (-1, 16.6+0.1) (X_label) {$X$};
	\node[] at (-0.5, 16.6+0.1) (XX_label) {$\mathcal{X}$};
	\node[] at (0.5, 16.6+0.1) (YY_label) {$\mathcal{Y}$};
	\node[] at (1, 16.6+0.1) (Y_label) {$Y$};
	
	\node[draw,circle,fill,scale=0.5] at (6, 3) (a_1_1) {};
	\node[draw,circle,fill,scale=0.5] at (6, 4) (a_1_2) {};
	\draw \ifcoltwo [olive,thick] \else [] \fi (6,3) -- (6.5,3) -- (6.5,0) -- (-5,0) -- (-5,6) -- (-4,6);
	\draw \ifcoltwo [violet,thick] \else [] \fi (6,4) -- (7,4) -- (7,-1) -- (-5.5,-1) -- (-5.5,7) -- (-4,7);
	
	\node[draw,circle,fill,scale=0.5] at (6, 8) (a_2_1) {};
	\node[draw,circle,fill,scale=0.5] at (6, 9) (a_2_2) {};
	\draw \ifcoltwo [red,thick] \else [] \fi (6,8) -- (8,8) -- (8,18.5) -- (-6,18.5) -- (-6,10) -- (-4,10);
	\draw \ifcoltwo [green,thick] \else [] \fi (6,9) -- (7.5,9) -- (7.5,18) -- (-5.5,18) -- (-5.5,11) -- (-4,11);
	
	\node[draw,circle,fill,scale=0.5] at (6, 13) (a_3_1) {};
	\node[draw,circle,fill,scale=0.5] at (6, 14) (a_3_2) {};
	\draw \ifcoltwo [orange,thick] \else [] \fi (6,13) -- (7,13) -- (7,17.5) -- (-5,17.5) -- (-5,14) -- (-4,14);
	\draw \ifcoltwo [blue,thick] \else [] \fi (6,14) -- (6.5,14) -- (6.5,17) -- (-4.5,17) -- (-4.5,15) -- (-4,15);
	
	\node[draw,circle,fill,scale=0.5] at (4, 2.5) (a_1_i1) {};
	\node[draw,circle,fill,scale=0.5] at (4, 4.5) (a_1_i2) {};
	\node[draw,circle,fill,scale=0.5] at (6, 3) (a_1_o1) {};
	\node[draw,circle,fill,scale=0.5] at (6, 4) (a_1_o2) {};
	\node[draw,rounded corners,fit=(a_1_i1) (a_1_i2) (a_1_o1) (a_1_o2)] (a_1_gadget) {};
	\node[] at (5, 3.5) (a_1_gadget_label) {$C_1$};
	
	\node[draw,circle,fill,scale=0.5] at (4, 7.5) (a_2_i1) {};
	\node[draw,circle,fill,scale=0.5] at (4, 8) (a_2_i2) {};
	\node[draw,circle,fill,scale=0.5] at (4, 9) (a_2_i3) {};
	\node[draw,circle,fill,scale=0.5] at (4, 9.5) (a_2_i4) {};
	\node[draw,circle,fill,scale=0.5] at (6, 8) (a_2_o1) {};
	\node[draw,circle,fill,scale=0.5] at (6, 9) (a_2_o2) {};
	\node[draw,rounded corners,fit=(a_2_i1) (a_2_i4) (a_2_o1) (a_2_o2)] (a_2_gadget) {};
	\node[] at (5, 8.5) (a_1_gadget_label) {$C_2$};
	
	\node[draw,circle,fill,scale=0.5] at (4, 12.5) (a_3_i1) {};
	\node[draw,circle,fill,scale=0.5] at (4, 12.75) (a_3_i2) {};
	\node[draw,circle,fill,scale=0.5] at (4, 13.1) (a_3_i3) {};
	\node[draw,circle,fill,scale=0.5] at (4, 13.35) (a_3_i4) {};
	\node[draw,circle,fill,scale=0.5] at (4, 13.65) (a_3_i5) {};
	\node[draw,circle,fill,scale=0.5] at (4, 13.9) (a_3_i6) {};
	\node[draw,circle,fill,scale=0.5] at (4, 14.25) (a_3_i7) {};
	\node[draw,circle,fill,scale=0.5] at (4, 14.5) (a_3_i8) {};
	\node[draw,circle,fill,scale=0.5] at (6, 13) (a_3_o1) {};
	\node[draw,circle,fill,scale=0.5] at (6, 14) (a_3_o2) {};
	\node[draw,rounded corners,fit=(a_3_i1) (a_3_i8) (a_3_o1) (a_3_o2)] (a_3_gadget) {};
	\node[] at (5, 13.5) (a_1_gadget_label) {$C_3$};
	
	\foreach \i in {1,...,16}{
		\node[draw,circle,fill,scale=0.5] at (-1, \i) (x_\i) {};
		\node[draw,circle,fill,scale=0.5] at (1, \i) (y_\i) {};
		\foreach \j in {0,...,3}{
			\node[draw,circle,fill,scale=0.5] at (-0.5, \i+\j*0.2-0.3) (x_\i_\j) {};
			\node[draw,circle,fill,scale=0.5] at (0.5, \i+\j*0.2-0.3) (y_\i_\j) {};
			\draw[] (x_\i) -- (x_\i_\j);
			\draw[] (y_\i) -- (y_\i_\j);
		}
		\foreach \j in {0,...,3}{
			\foreach \k in {0,...,3}{
				\draw[] (x_\i_\k) -- (y_\i_\j);
			}
		}
	}
	\foreach \i in {1,...,8}{
		\draw \ifcol [olive] \else [] \fi (x_\i) -- (v_2);
	}
	\foreach \i in {9,...,16}{
		\draw \ifcol [violet] \else [] \fi (x_\i) -- (v_3);
	}
	
	\begin{pgfonlayer}{bg}
		\draw \ifcol [olive] \else [] \fi (y_1) -- (a_1_i1);
		\draw \ifcol [olive] \else [] \fi (y_2) -- (a_1_i1);
		\draw \ifcol [olive] \else [] \fi (y_3) -- (a_1_i1);
		\draw \ifcol [olive] \else [] \fi (y_4) -- (a_1_i1);
		\draw \ifcol [olive] \else [] \fi (y_5) -- (a_1_i1);
		\draw \ifcol [olive] \else [] \fi (y_6) -- (a_1_i1);
		\draw \ifcol [olive] \else [] \fi (y_7) -- (a_1_i1);
		\draw \ifcol [olive] \else [] \fi (y_8) -- (a_1_i1);
		\draw \ifcol [violet] \else [] \fi (y_9) -- (a_1_i2);
		\draw \ifcol [violet] \else [] \fi (y_10) -- (a_1_i2);
		\draw \ifcol [violet] \else [] \fi (y_11) -- (a_1_i2);
		\draw \ifcol [violet] \else [] \fi (y_12) -- (a_1_i2);
		\draw \ifcol [violet] \else [] \fi (y_13) -- (a_1_i2);
		\draw \ifcol [violet] \else [] \fi (y_14) -- (a_1_i2);
		\draw \ifcol [violet] \else [] \fi (y_15) -- (a_1_i2);
		\draw \ifcol [violet] \else [] \fi (y_16) -- (a_1_i2);
		
		\draw \ifcol [red,semithick] \else [] \fi (y_1) -- (a_2_i1);
		\draw \ifcol [red,semithick] \else [] \fi (y_2) -- (a_2_i1);
		\draw \ifcol [red,semithick] \else [] \fi (y_3) -- (a_2_i1);
		\draw \ifcol [red,semithick] \else [] \fi (y_4) -- (a_2_i1);
		\draw \ifcol [green,semithick] \else [] \fi (y_5) -- (a_2_i2);
		\draw \ifcol [green,semithick] \else [] \fi (y_6) -- (a_2_i2);
		\draw \ifcol [green,semithick] \else [] \fi (y_7) -- (a_2_i2);
		\draw \ifcol [green,semithick] \else [] \fi (y_8) -- (a_2_i2);
		\draw \ifcol [red,semithick] \else [] \fi (y_9) -- (a_2_i3);
		\draw \ifcol [red,semithick] \else [] \fi (y_10) -- (a_2_i3);
		\draw \ifcol [red,semithick] \else [] \fi (y_11) -- (a_2_i3);
		\draw \ifcol [red,semithick] \else [] \fi (y_12) -- (a_2_i3);
		\draw \ifcol [green,semithick] \else [] \fi (y_13) -- (a_2_i4);
		\draw \ifcol [green,semithick] \else [] \fi (y_14) -- (a_2_i4);
		\draw \ifcol [green,semithick] \else [] \fi (y_15) -- (a_2_i4);
		\draw \ifcol [green,semithick] \else [] \fi (y_16) -- (a_2_i4);
		
		\draw \ifcol [orange,thick] \else [] \fi (y_1) -- (a_3_i1);
		\draw \ifcol [orange,thick] \else [] \fi (y_2) -- (a_3_i1);
		\draw \ifcol [blue,thick] \else [] \fi (y_3) -- (a_3_i2);
		\draw \ifcol [blue,thick] \else [] \fi (y_4) -- (a_3_i2);
		\draw \ifcol [orange,thick] \else [] \fi (y_5) -- (a_3_i3);
		\draw \ifcol [orange,thick] \else [] \fi (y_6) -- (a_3_i3);
		\draw \ifcol [blue,thick] \else [] \fi (y_7) -- (a_3_i4);
		\draw \ifcol [blue,thick] \else [] \fi (y_8) -- (a_3_i4);
		\draw \ifcol [orange,thick] \else [] \fi (y_9) -- (a_3_i5);
		\draw \ifcol [orange,thick] \else [] \fi (y_10) -- (a_3_i5);
		\draw \ifcol [blue,thick] \else [] \fi (y_11) -- (a_3_i6);
		\draw \ifcol [blue,thick] \else [] \fi (y_12) -- (a_3_i6);
		\draw \ifcol [orange,thick] \else [] \fi (y_13) -- (a_3_i7);
		\draw \ifcol [orange,thick] \else [] \fi (y_14) -- (a_3_i7);
		\draw \ifcol [blue,thick] \else [] \fi (y_15) -- (a_3_i8);
		\draw \ifcol [blue,thick] \else [] \fi (y_16) -- (a_3_i8);
		
		\draw \ifcol [cyan,thick] \else [] \fi (x_1) -- (-4,14);
		\draw \ifcol [pink,thick] \else [] \fi (x_2) -- (-4,15);
		\draw \ifcol [cyan,thick] \else [] \fi (x_3) -- (-4,14);
		\draw \ifcol [pink,thick] \else [] \fi (x_4) -- (-4,15);
		\draw \ifcol [cyan,thick] \else [] \fi (x_5) -- (-4,14);
		\draw \ifcol [pink,thick] \else [] \fi (x_6) -- (-4,15);
		\draw \ifcol [cyan,thick] \else [] \fi (x_7) -- (-4,14);
		\draw \ifcol [pink,thick] \else [] \fi (x_8) -- (-4,15);
		\draw \ifcol [cyan,thick] \else [] \fi (x_9) -- (-4,14);
		\draw \ifcol [pink,thick] \else [] \fi (x_10) -- (-4,15);
		\draw \ifcol [cyan,thick] \else [] \fi (x_11) -- (-4,14);
		\draw \ifcol [pink,thick] \else [] \fi (x_12) -- (-4,15);
		\draw \ifcol [cyan,thick] \else [] \fi (x_13) -- (-4,14);
		\draw \ifcol [pink,thick] \else [] \fi (x_14) -- (-4,15);
		\draw \ifcol [cyan,thick] \else [] \fi (x_15) -- (-4,14);
		\draw \ifcol [pink,thick] \else [] \fi (x_16) -- (-4,15);
		
		\draw \ifcol [red,semithick] \else [] \fi (x_1) -- (-4,6);
		\draw \ifcol [red,semithick] \else [] \fi (x_2) -- (-4,6);
		\draw \ifcol [red,semithick] \else [] \fi (x_3) -- (-4,6);
		\draw \ifcol [red,semithick] \else [] \fi (x_4) -- (-4,6);
		\draw \ifcol [green,semithick] \else [] \fi (x_5) -- (-4,7);
		\draw \ifcol [green,semithick] \else [] \fi (x_6) -- (-4,7);
		\draw \ifcol [green,semithick] \else [] \fi (x_7) -- (-4,7);
		\draw \ifcol [green,semithick] \else [] \fi (x_8) -- (-4,7);
		\draw \ifcol [red,semithick] \else [] \fi (x_9) -- (-4,6);
		\draw \ifcol [red,semithick] \else [] \fi (x_10) -- (-4,6);
		\draw \ifcol [red,semithick] \else [] \fi (x_11) -- (-4,6);
		\draw \ifcol [red,semithick] \else [] \fi (x_12) -- (-4,6);
		\draw \ifcol [green,semithick] \else [] \fi (x_13) -- (-4,7);
		\draw \ifcol [green,semithick] \else [] \fi (x_14) -- (-4,7);
		\draw \ifcol [green,semithick] \else [] \fi (x_15) -- (-4,7);
		\draw \ifcol [green,semithick] \else [] \fi (x_16) -- (-3,7);
		
		\draw \ifcol [orange,thick] \else [] \fi (x_1) -- (-4,10);
		\draw \ifcol [orange,thick] \else [] \fi (x_2) -- (-4,10);
		\draw \ifcol [blue,thick] \else [] \fi (x_3) -- (-4,11);
		\draw \ifcol [blue,thick] \else [] \fi (x_4) -- (-4,11);
		\draw \ifcol [orange,thick] \else [] \fi (x_5) -- (-4,10);
		\draw \ifcol [orange,thick] \else [] \fi (x_6) -- (-4,10);
		\draw \ifcol [blue,thick] \else [] \fi (x_7) -- (-4,11);
		\draw \ifcol [blue,thick] \else [] \fi (x_8) -- (-4,11);
		\draw \ifcol [orange,thick] \else [] \fi (x_9) -- (-4,10);
		\draw \ifcol [orange,thick] \else [] \fi (x_10) -- (-4,10);
		\draw \ifcol [blue,thick] \else [] \fi (x_11) -- (-4,11);
		\draw \ifcol [blue,thick] \else [] \fi (x_12) -- (-4,11);
		\draw \ifcol [orange,thick] \else [] \fi (x_13) -- (-4,10);
		\draw \ifcol [orange,thick] \else [] \fi (x_14) -- (-4,10);
		\draw \ifcol [blue,thick] \else [] \fi (x_15) -- (-4,11);
		\draw \ifcol [blue,thick] \else [] \fi (x_16) -- (-4,11);
	\end{pgfonlayer}
	\end{tikzpicture}
	\caption[]{A concealer graph from the class $\mathcal{G}_4$.}
	\label{fig:G_4_competitiveRatio}
\end{figure}

Towards this goal, we first define the class of \emph{concealer graphs}, which we denote with $\mathcal{G}_k$ ($k\in \mathbb{N}$).
Concealer graphs resemble the graphs of the lower bound construction in \cite{tightLowerBound} closely.
Essentially, we swap out $\AND_i$ gadgets in the original construction for concealer gadgets $C_i$.
A concealer graph of $\mathcal{G}_4$ is illustrated in Figure~\ref{fig:G_4_competitiveRatio}.

The main idea is that we can then speed-up or slow-down particular strategies by changing the position of the correct pairs within the concealer gadgets. This forces one strategy to extensively search for the correct pairs, while another strategy finds them immediately.

In the rest of this section we provide formal arguments for the above claims.
We start with a precise description of concealer graphs. 
Then, we show that for every concealer graph there exists a fast strategy.
Contrarily, we then provide a slow concealer graph for every strategy.
Together these two statements prove Theorem~\ref{thm:competitive}.

\subsection{Concealer Graphs} 
The first ingredient for the concealer graphs is a ``splitting scheme'' that results in the worst case running time of $\Omega(m \log(n))$.
Consider a vertex set of size $n=2^k$, on which the following refinements are performed. First, we split the set in halves, then quarters, then eighths and so on, until all vertices have their own distinct color.
This gives us $\log(n)$ rounds of refinements, each with a cost of $\Omega(n)$.
This results in total costs of $\Omega(n \log(n))$.
By ensuring that sufficiently many edges are involved, the running time can be increased to $\Omega(m \log(n))$.

Concealer graphs can be used to cause the splitting scheme just described. The graphs contain \emph{middle layers} ($X,\mathcal{X},\mathcal{Y},Y$) (see Figure~\ref{fig:G_4_competitiveRatio}) in which the splitting scheme can be forced.
The graph is constructed in a way such that splitting $Y$ into halves, quarters, eighths and so on, causes the next halving refinement on $X$.
The edge colors in Figure~\ref{fig:G_4_competitiveRatio} indicate the splitting scheme. While the halves (yellow and purple) of $Y$ lead to a split of $X$ into quarters (red and green), the quarters of $Y$ lead to eighths (blue and orange) of $X$ and so on.
By initially splitting $X$ in halves, any color refinement algorithm needs to cycle through these layers until $X$ is fully discrete.

The core idea of the general lower bound construction in~\cite{tightLowerBound} is that the $\AND_i$ gadget enforces refinements with respect to \emph{every} block of level~$i$, which in turn ensures costs of $2^k \cdot k^2 \in \Omega(m)$ for every level.

We modify the construction to suit our purposes as follows. In the concealer graphs, we swap for each~$i$ the $\AND_i$ gadget for a concealer gadget~$C_i$.
On a particular graph, the worst case behavior is therefore not enforced for all refinement strategies anymore.
However, a deterministic online algorithm cannot choose for \emph{all} possible concealer gadgets the correct pair in level~$i$ to allow it to continue with level~$i+1$.
Hence, an adversary can construct a graph that makes a specific color refinement slow, while keeping a ``shortcut'' for other algorithms that choose the correct pair directly.

We now formally define the  class $\mathcal{G}_k$ of concealer graphs.
Note that for every $k \in \mathbb{N}$, we define a set of graphs $\mathcal{G}_k$.
Essentially, we describe a graph $G_k \in \mathcal{G}_k$ based on concealer gadgets, and the set $\mathcal{G}_k$ then simply consists of all possible instantiations (i.e., positions of the correct pairs) for the included concealer gadgets.

At its core, a graph $G_k \in \mathcal{G}_k$ consists of the four middle layers of vertices $(X, \mathcal{X}, \mathcal{Y}, Y)$, that are interconnected using additional gadgets.
Formally, the vertex set of $G_k$ includes $X= \{x_0,...,x_{2^k-1}\}$, $\mathcal{X}= \{x_i^j \ | \ 0 \leq i < 2^k, 0 \leq j < k\}$, $\mathcal{Y}= \{y_i^j \ | \ 0 \leq i < 2^k, 0 \leq j < k\}$, $Y= \{y_0,...,y_{2^k-1}\}$, a simple starting gadget induced by only three vertices $v_1,v_2,v_3$ and $k-1$ concealer gadgets.
For $0 \leq l \leq k$ and $0 \leq q \leq 2^l-1$ let $\mathcal{B}_q^l = \{ q2^{k-l},...,(q+1)2^{k-l}-1 \}$ be the $q$-th binary block of level $l$. 
We use this notation on all sets of size $2^k$ for some $k \in \mathbb{N}$. 

Every $x_i$ is connected to a corresponding $y_i$ via a complete bipartite graph of size $k$ consisting of vertices in $\mathcal{X}$ and $\mathcal{Y}$ (see Figure~\ref{fig:G_4_competitiveRatio}).
Formally, each $x_i$ is connected to all $x_i^j$, $y_i$ to all $y_i^j$ and $x_i^j$ to all $y_i^{j'}$. 
For each level $l \in \{1,...,k-1\}$, the i-th binary block of level~$l$ is connected to the $i$-th in-vertex of the $l$-th concealer gadget. 
Furthermore, for each gadget $C_l$, we connect $a_0$ to all $X_i^l$ with $i$ even and $a_1$ to all $X_i^l$ with $i$ odd. 
The starting construction splits $X$ into the blocks $X_0^0$ and $X_1^0$.
We refer to the $i$-th in-vertex of the $l$-th concealer gadget as ${b}_i^l$ and to the $i$-th out-vertex as $a_i^l$.

Let us generally consider how a refinement strategy has to operate on $G_k$.
The algorithm starts with the monochromatic coloring of $G_k$.
The first refinement always distinguishes vertices by their degree, meaning we get the individualized starting gadget $\{v_1\},\{v_2\},\{v_3\}$, the distinct layers in the middle $X,\mathcal{X}\cup\mathcal{Y},Y$, the in- and out-vertices of the concealer gadgets $\bigcup_{l\in\{1,...,k-1\}}\{b_i^l, a_j^l \ | \ i\in\{0,...,2^l\},j\in\{0,1\} \}$, and the union of the inner vertices of the concealer gadgets. Next the middle layers are split in half. 
From this point onwards the splits that are possible depend on finding the correct pair in the gadgets. This can lead to fast or slow refinements, as discussed next.

\subsection{A Fast Strategy for Every Concealer Graph} \label{subsec:fastconcealer}
We now show that for every \emph{fixed} concealer graph $G_k \in \mathcal{G}_k$ we can define a linear time strategy.
We show this by providing an appropriate sequence of refinements.

For each concealer gadget~$C_l$ in $G_k$, let ${b}_{i_l}^l,{b}_{i_l+1}^l$ be the correct pair.
Now consider an online refinement strategy on such a graph.
After the first (and fixed) refinement, we refine $X$ with respect to $\{v_2\}$ or $\{v_3\}$.
We choose one half of $X_{i_1}^1$ for the next refinement and then $\mathcal{X}_{i_1}^1$, $\mathcal{Y}_{i_1}^1$ and $Y_{i_1}^1$ while propagating the split through the middle layers.
The important property is that $Y_{i_1}^1$  always splits the correct pair of the next concealer gadget. 
The concealer gadget then in turn splits $X$ into quarters.
Now, we continue with the quarters $X_{i_2}^2$, $\mathcal{X}_{i_2}^2$, $\mathcal{Y}_{i_2}^2$ and $Y_{i_2}^2$, such that the second concealer gadget is activated.
This splits $X$ in eighths.

We now repeat this scheme, such that for each level we only propagate the blocks corresponding to correct pairs through the layers and immediately continue with the next level after activating the concealer gadget.
When $X$ is discrete, we get the equitable coloring by refining with respect to each level~$k$ block of $X$, $\mathcal{X}$, $\mathcal{Y}$ and $Y$.

Now consider the cost of this strategy. 
While cycling through the layers, the most expensive refinements are those with respect to the blocks of $\mathcal{X}$ and $\mathcal{Y}$.
On level~$l$, they have cost $2^{k-l} \cdot k^2$, which means the total cost for all levels is $2^{k} \cdot k^2 = \Theta(m)$.
Once~$X$ is discrete the cost of the final refinements of $\mathcal{X}$, $\mathcal{Y}$ and $Y$ is also in~$\Theta(m)$.

Overall, the cost for an optimal solution for~$G_k$ is linear, i.e., $\opt(G_k) \in \Theta(m)$.
Note that since refinement is always continued with color classes that have just been created, the scheme actually follows a depth-first approach and can be implemented using a stack.  

\subsection{A Slow Concealer Graph for Every Strategy}  \label{subsec:slowconcealer}
For a \emph{fixed} strategy $W$, we now provide an infinite family of concealer graphs $G_k$ on which this strategy is slow, i.e., incurs super-linear cost.
The family is constructed by choosing for every $k \in \mathbb{N}$ one specific concealer graph $G_k \in \mathcal{G}_k$. 

We start with an arbitrary graph $G_k \in \mathcal{G}_k$.
We run $W$ on $G_k$ and observe which color classes are split within the concealer gadgets.
Say we are looking at concealer gadget $C_i$. 
If $W$ distinguishes the correct pair in~$G_k$, but there are still dead ends that have not been distinguished, then we replace~$G_k$ by the graph~$G'_k\in \mathcal{G}_k$ obtained from~$G_k$ by replacing the gadget~$C_i$ with another one so that a dead end not yet investigated becomes the correct pair.
Due to Lemma~\ref{lem:concealer} and Lemma~\ref{lem:concealer:refine} we know that up until the point where $W$ finds the correct pair in~$C_i$ for graph~$G_k$, the strategy $W$ performs the same sequence of splits when executed on~$G'_k$ as on~$G_k$.
Thus, by doing these transformations exhaustively, we ensure $W$ distinguishes all correct pairs in all the concealer gadgets last.
This causes $2^k \cdot k^2$ cost per level and hence $2^k \cdot k^3 = \Theta(m \log(n))$ total cost. 

Since the optimal solution for fixed $G_k$ only has linear cost (see Section~\ref{subsec:fastconcealer}), we in turn get that  
$\cost(W, G_k) \in \Omega(\opt(G_k) \cdot \log(\opt(G_k)))$. 
This in turn proves Theorem~\ref{thm:competitive}.
Further details and a more formal reasoning can be found in Appendix~\ref{subsec:slowconcealer_appendix}.
\section{Comparison of Practical Worklists}
We now compare specific, practical worklist data structures.
First, we compare stacks with queues. We show that either of the two can asymptotically outperform the other.
Note that it is also possible to show the same result for priority queues (see Appendix~\ref{sec:priorityQueues}).

\subsection{Stack Advantage over Queue} \label{sec:Stack>Queue}
To see how a stack worklist might outperform a queue, 
recall the fast strategies for concealer graphs of Section~\ref{sec:competitiveRatio}.
The specific, fast split scheme discussed there is realized by a worklist maintained as a stack. Indeed, whenever possible we continue with a ``newest'' class. 

We conclude from Theorem~\ref{thm:competitive:ratio} that there is a class of graphs on which a stack based worklist asymptotically outperforms a queue based worklist by a logarithmic factor.

We should remark that it is possible to prove the same result with a simpler construction that does not rely on concealer gadgets. 
We should also remark that the construction does not apply to all stack based worklists. However, it is possible to modify the construction such that a particular stack based worklist is optimal. 
For example this can be done for the worklist that choses smallest color classes first (see Appendix~\ref{sec:Stack>Queue_appendix}).

\subsection{Queue Advantage over Stack}
\label{sec:Queue>Stack}

Now, we construct a graph class, called the \emph{queue graphs}, for which a queue based worklist outperforms a stack based one by a logarithmic factor.
This complements the result of the previous section.
The construction is also based on the graph class of Berkholz et al.~\cite{tightLowerBound}.
It is an extension of these graphs, which allows queue worklists to finish quickly but maintains the slow behavior for stacks. We provide an intuitive description. A formal definition and a detailed analysis is given in Appendix~\ref{sec:Queue>Stack_appendix}.

\textbf{The starting gadget.}
We use a starting gadget (see Figure~\ref{fig:startingGadget_Queue>Stack}) that forces a stack worklist to perform certain splits before others, while a queue worklist behaves differently. 

Consider the gadget together with the coloring indicated in the figure.
Any color refinement eventually splits the pairs $\{p_1^{1,3},p_1^{2,3}\}$, $\{p_2^{1,3}, p_2^{2,3}\}$ and $\{a_0,a_1\}$.
However, a stack based worklist splits $\{p_1^{1,3},p_1^{2,3}\}$ or $\{p_2^{1,3}, p_2^{2,3}\}$ \emph{before} $\{a_0,a_1\}$, while a queue based one splits $\{a_0,a_1\}$ before the other two pairs.

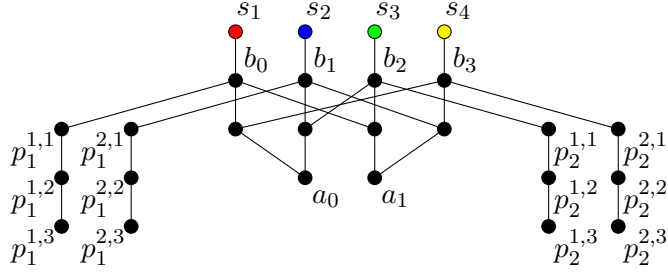
\begin{figure}
	\centering
	\begin{tikzpicture}[rotate=270,scale=0.925]
	\node[draw,circle,fill=red,scale=0.5] at (-0.7,0) (s_0) {};
	\node[draw,circle,fill=blue,scale=0.5] at (-0.7,1) (s_1) {};
	\node[draw,circle,fill=green,scale=0.5] at (-0.7,2) (s_2) {};
	\node[draw,circle,fill=yellow,scale=0.5] at (-0.7,3) (s_3) {};
	
	\node[draw,circle,fill,scale=0.5] at (0.7,-2.5) (p_1_1_1) {};
	\node[draw,circle,fill,scale=0.5] at (0.7,-1.5) (p_1_2_1) {};
	\node[draw,circle,fill,scale=0.5] at (0.7,4.5) (p_2_1_1) {};
	\node[draw,circle,fill,scale=0.5] at (0.7,5.5) (p_2_2_1) {};
	
	\node[draw,circle,fill,scale=0.5] at (1.4,-2.5) (p_1_1_2) {};
	\node[draw,circle,fill,scale=0.5] at (1.4,-1.5) (p_1_2_2) {};
	\node[draw,circle,fill,scale=0.5] at (1.4,4.5) (p_2_1_2) {};
	\node[draw,circle,fill,scale=0.5] at (1.4,5.5) (p_2_2_2) {};
	
	\node[draw,circle,fill,scale=0.5] at (2.1,-2.5) (p_1_1_3) {};
	\node[draw,circle,fill,scale=0.5] at (2.1,-1.5) (p_1_2_3) {};
	\node[draw,circle,fill,scale=0.5] at (2.1,4.5) (p_2_1_3) {};
	\node[draw,circle,fill,scale=0.5] at (2.1,5.5) (p_2_2_3) {};
	
	\node[draw,circle,fill,scale=0.5] at (0,0) (b_0) {};
	\node[draw,circle,fill,scale=0.5] at (0,1) (b_1) {};
	\node[draw,circle,fill,scale=0.5] at (0,2) (b_2) {};
	\node[draw,circle,fill,scale=0.5] at (0,3) (b_3) {};
	
	\node[draw,circle,fill,scale=0.5] at (0.7,0) (c_0) {};
	\node[draw,circle,fill,scale=0.5] at (0.7,1) (c_1) {};
	\node[draw,circle,fill,scale=0.5] at (0.7,2) (c_2) {};
	\node[draw,circle,fill,scale=0.5] at (0.7,3) (c_3) {};
	
	\node[draw,circle,fill,scale=0.5] at (1.4,1) (a_0) {};
	\node[draw,circle,fill,scale=0.5] at (1.4,2) (a_1) {};
	
	\node[] at (-0.3,0.3) (b_0_label) {$b_0$};
	\node[] at (-0.3,1.3) (b_1_label) {$b_1$};
	\node[] at (-0.3,2.3) (b_2_label) {$b_2$};
	\node[] at (-0.3,3.3) (b_3_label) {$b_3$};
	
	\node[] at (1.7,1.3) (a_0_label) {$a_0$};
	\node[] at (1.7,2.3) (a_1_label) {$a_1$};
	
	\node[] at (-1,0.2) (s_1_label) {$s_1$};
	\node[] at (-1,1.2) (s_2_label) {$s_2$};
	\node[] at (-1,2.2) (s_3_label) {$s_3$};
	\node[] at (-1,3.2) (s_4_label) {$s_4$};
	
	\node[] at (1,-2.9) (p_0_1_label) {$p_1^{1,1}$};
	\node[] at (1,-1.9) (p_1_1_label) {$p_1^{2,1}$};
	\node[] at (1,4.9) (p_2_1_label) {$p_2^{1,1}$};
	\node[] at (1,5.9) (p_3_1_label) {$p_2^{2,1}$};
	
	\node[] at (1.7,-2.9) (p_0_2_label) {$p_1^{1,2}$};
	\node[] at (1.7,-1.9) (p_1_2_label) {$p_1^{2,2}$};
	\node[] at (1.7,4.9) (p_2_2_label) {$p_2^{1,2}$};
	\node[] at (1.7,5.9) (p_3_2_label) {$p_2^{2,2}$};
	
	\node[] at (2.4,-2.9) (p_0_3_label) {$p_1^{1,3}$};
	\node[] at (2.4,-1.9) (p_1_3_label) {$p_1^{2,3}$};
	\node[] at (2.4,4.9) (p_2_3_label) {$p_2^{1,3}$};
	\node[] at (2.4,5.9) (p_3_3_label) {$p_2^{2,3}$};
	
	\draw (b_0) -- (c_0);
	\draw (b_0) -- (c_2);
	\draw (b_1) -- (c_1);
	\draw (b_1) -- (c_3);
	\draw (b_2) -- (c_1);
	\draw (b_2) -- (c_2);
	\draw (b_3) -- (c_0);
	\draw (b_3) -- (c_3);
	
	\draw (c_0) -- (a_0);
	\draw (c_1) -- (a_0);
	\draw (c_2) -- (a_1);
	\draw (c_3) -- (a_1);
	
	\draw (s_0) -- (b_0);
	\draw (s_1) -- (b_1);
	\draw (s_2) -- (b_2);
	\draw (s_3) -- (b_3);
	
	\draw (p_1_1_1) -- (b_0);
	\draw (p_1_2_1) -- (b_1);
	\draw (p_2_1_1) -- (b_2);
	\draw (p_2_2_1) -- (b_3);
	
	\draw (p_1_1_1) -- (p_1_1_2);
	\draw (p_1_2_1) -- (p_1_2_2);
	\draw (p_2_1_1) -- (p_2_1_2);
	\draw (p_2_2_1) -- (p_2_2_2);
	
	\draw (p_1_1_2) -- (p_1_1_3);
	\draw (p_1_2_2) -- (p_1_2_3);
	\draw (p_2_1_2) -- (p_2_1_3);
	\draw (p_2_2_2) -- (p_2_2_3);
	\end{tikzpicture}
	\caption{An extension of $\AND_2$ gadget, which will be used as the starting gadget of the new construction. The starting vertices $s_1,...,s_4$ have been individualized.}
	\label{fig:startingGadget_Queue>Stack}
\end{figure}

\textbf{Graph class construction.} 
We start with the graphs from~\cite{tightLowerBound} as a main building block.
Recall that these graphs are the graphs from Section~\ref{sec:competitiveRatio} 
where the concealer gadgets are replaced by $\AND$ gadgets. 
As argued in~\cite{tightLowerBound} a worst case behavior is enforced for every refinement strategy: any refinement on these graphs has a cost of $\Omega(2^k \cdot k^3)=\Omega(m \log(n))$.

We now add ``shortcuts'' that allow queue based algorithms to bypass the construction. 
The core idea is to ensure the queue algorithm refines the set $X$ into a discrete set within a single level of its breadth first behavior. 
This causes $\mathcal{X}$, $\mathcal{Y}$ and $Y$ to completely split in subsequent rounds, thereby preventing the cycling behavior that causes superlinear cost.
Indeed, if $\mathcal{X}$ and $\mathcal{Y}$ are handled only once, then the total cost is in $\mathcal{O}(2^k \cdot k^2) = \mathcal{O}(m)$.

Simultaneously we force the stack into the typical cycling behavior. 
We do so by forcing it to make the same splits of $X$ as the simple starting gadget from Section~\ref{sec:competitiveRatio} would.

We apply the following changes to define queue graphs $G^{\text{Q}}_k$ (see Figure~\ref{fig:G_3_Queue>Stack}): 
	  we connect each vertex in $X$ to a path of length $k$. 
	  We add the new starting gadget described above. 
	  We extend the paths $p_1$ and $p_2$ to a length of $k+2$ and connect the ends to the old starting vertices through unidirectional gadgets. 
	  We also attach a third path $p_{\text{Q}}$ of length $k$ to $a_0$ and $a_1$ and connect the i-th pair to the level-$i$ blocks of the $i$-th vertices of the $X$-paths, again through unidirectional gadgets.
Note that the graph has still a size of $\mathcal{O}(2^k \cdot k^2)$.

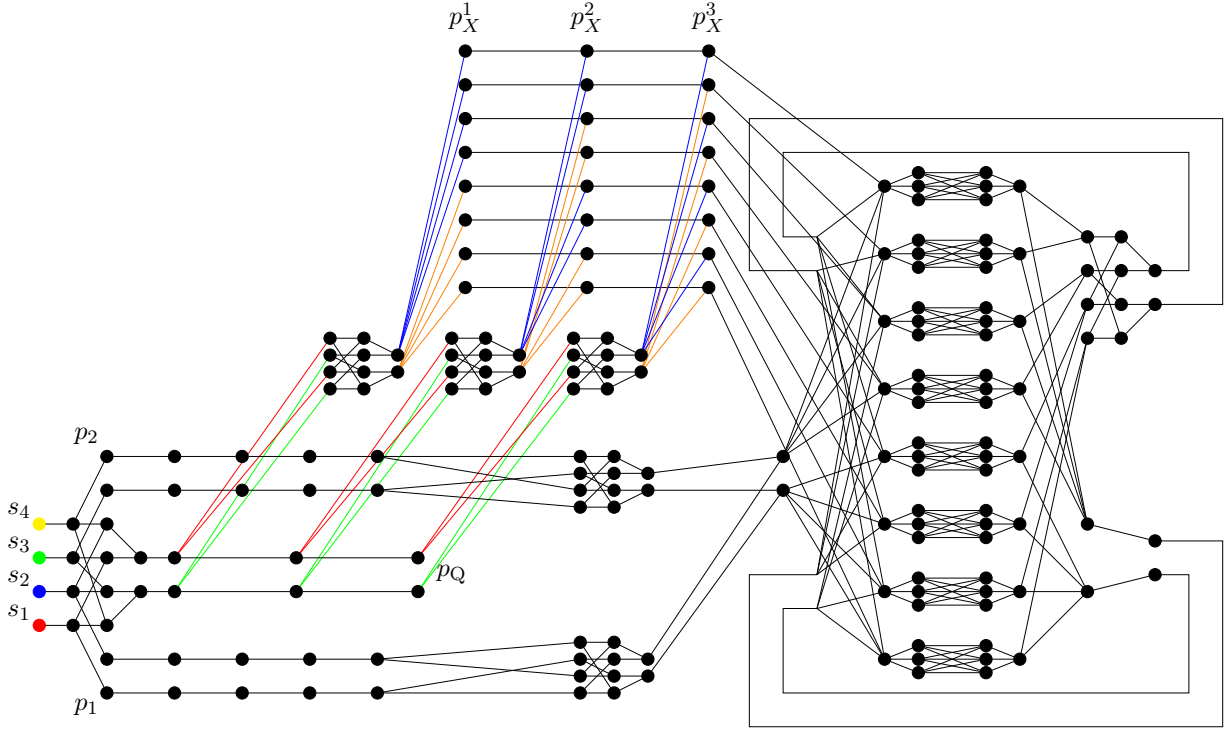
\begin{figure}[t]
	\centering
	\resizebox{1\textwidth}{!}
	{
		\begin{tikzpicture}[]
		
		\node[draw,circle,fill,scale=0.5] at (2,5.75) (b_2_0) {};
		\node[draw,circle,fill,scale=0.5] at (2,6.25) (b_2_1) {};
		\node[draw,circle,fill,scale=0.5] at (2,6.75) (b_2_2) {};
		\node[draw,circle,fill,scale=0.5] at (2,7.25) (b_2_3) {};
		
		\node[draw,circle,fill,scale=0.5] at (2.5,5.75) (c_2_0) {};
		\node[draw,circle,fill,scale=0.5] at (2.5,6.25) (c_2_1) {};
		\node[draw,circle,fill,scale=0.5] at (2.5,6.75) (c_2_2) {};
		\node[draw,circle,fill,scale=0.5] at (2.5,7.25) (c_2_3) {};
		
		\node[draw,circle,fill,scale=0.5] at (3,6.25) (a_2_0) {};
		\node[draw,circle,fill,scale=0.5] at (3,6.75) (a_2_1) {};
		
		\draw (b_2_0) -- (c_2_0);
		\draw (b_2_0) -- (c_2_2);
		\draw (b_2_1) -- (c_2_1);
		\draw (b_2_1) -- (c_2_3);
		\draw (b_2_2) -- (c_2_1);
		\draw (b_2_2) -- (c_2_2);
		\draw (b_2_3) -- (c_2_0);
		\draw (b_2_3) -- (c_2_3);
		
		\draw (c_2_0) -- (a_2_0);
		\draw (c_2_1) -- (a_2_0);
		\draw (c_2_2) -- (a_2_1);
		\draw (c_2_3) -- (a_2_1);
		
		\draw [] (3,6.25) -- (4,6.25) -- (4,9) -- (-3,9) -- (-3,6.75) -- (-2,6.75);
		\draw [] (3,6.75) -- (3.5,6.75) -- (3.5,8.5) -- (-2.5,8.5) -- (-2.5,7.25) -- (-2,7.25);

		\node[draw,circle,fill,scale=0.5] at (2, 2) (b_1_0) {};
		\node[draw,circle,fill,scale=0.5] at (2, 3) (b_1_1) {};
		\node[draw,circle,fill,scale=0.5] at (3, 2.25) (a_1_0) {};
		\node[draw,circle,fill,scale=0.5] at (3, 2.75) (a_1_1) {};
		
		\draw (b_1_0) -- (a_1_0);
		\draw (b_1_1) -- (a_1_1);
		
		\draw [] (3,2.25) -- (3.5,2.25) -- (3.5,0.5) -- (-2.5,0.5) -- (-2.5,1.75) -- (-2,1.75);
		\draw [] (3,2.75) -- (4,2.75) -- (4,0) -- (-3,0) -- (-3,2.25) -- (-2,2.25);
		
		\def \xStart{-13.5}
		\def \yStart{-1.5}
		\def \yStartXPath{-1.5}
		
		\node[draw,circle,fill,red,scale=0.5] at (\xStart,\yStart+3+0*0.5) (s_0) {};
		\node[draw,circle,fill,blue,scale=0.5] at (\xStart,\yStart+3+1*0.5) (s_1) {};
		\node[draw,circle,fill,green,scale=0.5] at (\xStart,\yStart+3+2*0.5) (s_2) {};
		\node[draw,circle,fill,yellow,scale=0.5] at (\xStart,\yStart+3+3*0.5) (s_3) {};
		\foreach \i in {0,...,3}{
			\node[draw,circle,fill,scale=0.5] at (\xStart+0.5,\yStart+3+\i*0.5) (b_s_\i) {};
			\node[draw,circle,fill,scale=0.5] at (\xStart+1,\yStart+3+\i*0.5) (c_s_\i) {};
		}
		\foreach \i in {1,...,5}{
			\node[draw,circle,fill,scale=0.5] at ({\xStart+1+(\i-1)*1},\yStart+2) (p_1_1_\i) {};
			\node[draw,circle,fill,scale=0.5] at ({\xStart+1+(\i-1)*1},\yStart+2.5) (p_1_2_\i) {};
			\node[draw,circle,fill,scale=0.5] at ({\xStart+1+(\i-1)*1},\yStart+5) (p_2_1_\i) {};
			\node[draw,circle,fill,scale=0.5] at ({\xStart+1+(\i-1)*1},\yStart+5.5) (p_2_2_\i) {};
		}
		\node[] at (\xStart+.7,\yStart+1.8) (p_1_label) {$p_1$};
		\node[] at (\xStart+.7,\yStart+5.8) (p_2_label) {$p_2$};
		\foreach \i/\j in {1/2,2/3,3/4,4/5}{
			\draw (p_1_1_\i) -- (p_1_1_\j);
			\draw (p_1_2_\i) -- (p_1_2_\j);
			\draw (p_2_1_\i) -- (p_2_1_\j);
			\draw (p_2_2_\i) -- (p_2_2_\j);
		}
		\node[draw,circle,fill,scale=0.5] at (\xStart+1.5,\yStart+3.5) (a_s_0) {};
		\node[draw,circle,fill,scale=0.5] at (\xStart+1.5,\yStart+4) (a_s_1) {};
		
		\node[] at (\xStart-.3,\yStart+3.2) (s_1_label) {$s_1$};
		\node[] at (\xStart-.3,\yStart+3.7) (s_2_label) {$s_2$};
		\node[] at (\xStart-.3,\yStart+4.2) (s_3_label) {$s_3$};
		\node[] at (\xStart-.3,\yStart+4.7) (s_4_label) {$s_4$};
		
		\draw (b_s_0) -- (c_s_0);
		\draw (b_s_0) -- (c_s_2);
		\draw (b_s_1) -- (c_s_1);
		\draw (b_s_1) -- (c_s_3);
		\draw (b_s_2) -- (c_s_1);
		\draw (b_s_2) -- (c_s_2);
		\draw (b_s_3) -- (c_s_0);
		\draw (b_s_3) -- (c_s_3);
		
		\draw (c_s_0) -- (a_s_0);
		\draw (c_s_1) -- (a_s_0);
		\draw (c_s_2) -- (a_s_1);
		\draw (c_s_3) -- (a_s_1);
		
		\draw (s_0) -- (b_s_0);
		\draw (s_1) -- (b_s_1);
		\draw (s_2) -- (b_s_2);
		\draw (s_3) -- (b_s_3);
		
		\draw (p_1_1_1) -- (b_s_0);
		\draw (p_1_2_1) -- (b_s_1);
		\draw (p_2_1_1) -- (b_s_2);
		\draw (p_2_2_1) -- (b_s_3);
		
		\node[draw,circle,fill,scale=0.5] at (\xStart+8,\yStart+2) (b_e1_0) {};
		\node[draw,circle,fill,scale=0.5] at (\xStart+8,\yStart+2.25) (b_e1_1) {};
		\node[draw,circle,fill,scale=0.5] at (\xStart+8,\yStart+2.5) (b_e1_2) {};
		\node[draw,circle,fill,scale=0.5] at (\xStart+8,\yStart+2.75) (b_e1_3) {};
		
		\node[draw,circle,fill,scale=0.5] at (\xStart+8.5,\yStart+2) (c_e1_0) {};
		\node[draw,circle,fill,scale=0.5] at (\xStart+8.5,\yStart+2.25) (c_e1_1) {};
		\node[draw,circle,fill,scale=0.5] at (\xStart+8.5,\yStart+2.5) (c_e1_2) {};
		\node[draw,circle,fill,scale=0.5] at (\xStart+8.5,\yStart+2.75) (c_e1_3) {};
		
		\node[draw,circle,fill,scale=0.5] at (\xStart+9,\yStart+2.25) (a_e1_0) {};
		\node[draw,circle,fill,scale=0.5] at (\xStart+9,\yStart+2.5) (a_e1_1) {};
		
		\draw (b_e1_0) -- (c_e1_0);
		\draw (b_e1_0) -- (c_e1_2);
		\draw (b_e1_1) -- (c_e1_1);
		\draw (b_e1_1) -- (c_e1_3);
		\draw (b_e1_2) -- (c_e1_1);
		\draw (b_e1_2) -- (c_e1_2);
		\draw (b_e1_3) -- (c_e1_0);
		\draw (b_e1_3) -- (c_e1_3);
		
		\draw (c_e1_0) -- (a_e1_0);
		\draw (c_e1_1) -- (a_e1_0);
		\draw (c_e1_2) -- (a_e1_1);
		\draw (c_e1_3) -- (a_e1_1);
		
		\node[draw,circle,fill,scale=0.5] at (\xStart+8,\yStart+4.75) (b_e2_0) {};
		\node[draw,circle,fill,scale=0.5] at (\xStart+8,\yStart+5) (b_e2_1) {};
		\node[draw,circle,fill,scale=0.5] at (\xStart+8,\yStart+5.25) (b_e2_2) {};
		\node[draw,circle,fill,scale=0.5] at (\xStart+8,\yStart+5.5) (b_e2_3) {};
		
		\node[draw,circle,fill,scale=0.5] at (\xStart+8.5,\yStart+4.75) (c_e2_0) {};
		\node[draw,circle,fill,scale=0.5] at (\xStart+8.5,\yStart+5) (c_e2_1) {};
		\node[draw,circle,fill,scale=0.5] at (\xStart+8.5,\yStart+5.25) (c_e2_2) {};
		\node[draw,circle,fill,scale=0.5] at (\xStart+8.5,\yStart+5.5) (c_e2_3) {};
		
		\node[draw,circle,fill,scale=0.5] at (\xStart+9,\yStart+5) (a_e2_0) {};
		\node[draw,circle,fill,scale=0.5] at (\xStart+9,\yStart+5.25) (a_e2_1) {};
		
		\draw (b_e2_0) -- (c_e2_0);
		\draw (b_e2_0) -- (c_e2_2);
		\draw (b_e2_1) -- (c_e2_1);
		\draw (b_e2_1) -- (c_e2_3);
		\draw (b_e2_2) -- (c_e2_1);
		\draw (b_e2_2) -- (c_e2_2);
		\draw (b_e2_3) -- (c_e2_0);
		\draw (b_e2_3) -- (c_e2_3);
		
		\draw (c_e2_0) -- (a_e2_0);
		\draw (c_e2_1) -- (a_e2_0);
		\draw (c_e2_2) -- (a_e2_1);
		\draw (c_e2_3) -- (a_e2_1);
		
		\node[draw,circle,fill,scale=0.5] at (\xStart+11,3.5) (pS_e_0) {};
		\node[draw,circle,fill,scale=0.5] at (\xStart+11,4) (pS_e_1) {};

		\draw (p_1_1_5) -- (b_e1_0);
		\draw (p_1_2_5) -- (b_e1_1);
		\draw (p_2_1_5) -- (b_e2_0);
		\draw (p_2_2_5) -- (b_e2_1);
		
		\draw (p_1_1_5) -- (b_e1_2);
		\draw (p_1_2_5) -- (b_e1_3);
		\draw (p_2_1_5) -- (b_e2_2);
		\draw (p_2_2_5) -- (b_e2_3);
		
		\draw (a_e1_0) -- (pS_e_0);
		\draw (a_e1_1) -- (pS_e_1);
		\draw (a_e2_0) -- (pS_e_0);
		\draw (a_e2_1) -- (pS_e_1);

		\foreach \j in {1,...,3}{
			\node[draw,circle,fill,scale=0.5] at (\xStart+2+\j*1.8-1.8,\yStart+3.5) (pQ_\j_1) {};
			\node[draw,circle,fill,scale=0.5] at (\xStart+2+\j*1.8-1.8,\yStart+4) (pQ_\j_2) {};
		}
		\draw[] (a_s_0) -- (pQ_1_1);
		\draw[] (a_s_1) -- (pQ_1_2);
		\draw[] (pQ_1_1) -- (pQ_2_1);
		\draw[] (pQ_1_2) -- (pQ_2_2);
		\draw[] (pQ_2_1) -- (pQ_3_1);
		\draw[] (pQ_2_2) -- (pQ_3_2);
		
		\node[] at (\xStart+2+2*1.8+0.5,\yStart+3.75) (p_q_label) {$p_{\text{Q}}$};
		
		\foreach \j in {1,...,3}{
			
			\node[draw,circle,fill,scale=0.5] at (\xStart+2.5+\j*1.8,\yStartXPath+6.5) (b_q\j_0) {};
			\node[draw,circle,fill,scale=0.5] at (\xStart+2.5+\j*1.8,\yStartXPath+6.75) (b_q\j_1) {};
			\node[draw,circle,fill,scale=0.5] at (\xStart+2.5+\j*1.8,\yStartXPath+7) (b_q\j_2) {};
			\node[draw,circle,fill,scale=0.5] at (\xStart+2.5+\j*1.8,\yStartXPath+7.25) (b_q\j_3) {};
			
			\node[draw,circle,fill,scale=0.5] at (\xStart+3+\j*1.8,\yStartXPath+6.5) (c_q\j_0) {};
			\node[draw,circle,fill,scale=0.5] at (\xStart+3+\j*1.8,\yStartXPath+6.75) (c_q\j_1) {};
			\node[draw,circle,fill,scale=0.5] at (\xStart+3+\j*1.8,\yStartXPath+7) (c_q\j_2) {};
			\node[draw,circle,fill,scale=0.5] at (\xStart+3+\j*1.8,\yStartXPath+7.25) (c_q\j_3) {};
			
			\node[draw,circle,fill,scale=0.5] at (\xStart+3.5+\j*1.8,\yStartXPath+6.75) (a_q\j_0) {};
			\node[draw,circle,fill,scale=0.5] at (\xStart+3.5+\j*1.8,\yStartXPath+7) (a_q\j_1) {};
			
			\draw (b_q\j_0) -- (c_q\j_0);
			\draw (b_q\j_0) -- (c_q\j_2);
			\draw (b_q\j_1) -- (c_q\j_1);
			\draw (b_q\j_1) -- (c_q\j_3);
			\draw (b_q\j_2) -- (c_q\j_1);
			\draw (b_q\j_2) -- (c_q\j_2);
			\draw (b_q\j_3) -- (c_q\j_0);
			\draw (b_q\j_3) -- (c_q\j_3);
			
			\draw (c_q\j_0) -- (a_q\j_0);
			\draw (c_q\j_1) -- (a_q\j_0);
			\draw (c_q\j_2) -- (a_q\j_1);
			\draw (c_q\j_3) -- (a_q\j_1);
			
		}
		\begin{pgfonlayer}{bg}
			\foreach \i in {1,...,3}{
				\draw [green] (pQ_\i_1) -- (b_q\i_0);
				\draw [green] (pQ_\i_1) -- (b_q\i_2);
				\draw [red] (pQ_\i_2) -- (b_q\i_1);
				\draw [red] (pQ_\i_2) -- (b_q\i_3);
			}
		\end{pgfonlayer}
		
		\foreach \i in {1,...,8}{
			\node[draw,circle,fill,scale=0.5] at (-1, \i) (x_\i) {};
			\node[draw,circle,fill,scale=0.5] at (1, \i) (y_\i) {};
			\foreach \j in {-1,...,1}{
				\node[draw,circle,fill,scale=0.5] at (-0.5, \i+\j*0.2) (x_\i_\j) {};
				\node[draw,circle,fill,scale=0.5] at (0.5, \i+\j*0.2) (y_\i_\j) {};
				\draw[] (x_\i) -- (x_\i_\j);
				\draw[] (y_\i) -- (y_\i_\j);
			}
			\foreach \j in {-1,...,1}{
				\foreach \k in {-1,...,1}{
					\draw[] (x_\i_\k) -- (y_\i_\j);
				}
			}
		}
	
		\foreach \i in {1,...,4}{
			\draw [] (x_\i) -- (pS_e_0);
		}
		\foreach \i in {5,...,8}{
			\draw [] (x_\i) -- (pS_e_1);
		}
		
		\draw [] (y_1) -- (b_1_0);
		\draw [] (y_2) -- (b_1_0);
		\draw [] (y_3) -- (b_1_0);
		\draw [] (y_4) -- (b_1_0);
		\draw [] (y_5) -- (b_1_1);
		\draw [] (y_6) -- (b_1_1);
		\draw [] (y_7) -- (b_1_1);
		\draw [] (y_8) -- (b_1_1);
		
		\draw [] (y_1) -- (b_2_0);
		\draw [] (y_2) -- (b_2_0);
		\draw [] (y_3) -- (b_2_1);
		\draw [] (y_4) -- (b_2_1);
		\draw [] (y_5) -- (b_2_2);
		\draw [] (y_6) -- (b_2_2);
		\draw [] (y_7) -- (b_2_3);
		\draw [] (y_8) -- (b_2_3);
		
		\draw [] (x_1) -- (-2,1.75);
		\draw [] (x_2) -- (-2,1.75);
		\draw [] (x_3) -- (-2,2.25);
		\draw [] (x_4) -- (-2,2.25);
		\draw [] (x_5) -- (-2,1.75);
		\draw [] (x_6) -- (-2,1.75);
		\draw [] (x_7) -- (-2,2.25);
		\draw [] (x_8) -- (-2,2.25);
		
		\draw [] (x_1) -- (-2,6.75);
		\draw [] (x_2) -- (-2,7.25);
		\draw [] (x_3) -- (-2,6.75);
		\draw [] (x_4) -- (-2,7.25);
		\draw [] (x_5) -- (-2,6.75);
		\draw [] (x_6) -- (-2,7.25);
		\draw [] (x_7) -- (-2,6.75);
		\draw [] (x_8) -- (-2,7.25);
		
		\foreach \i in {1,...,8}{
			\foreach \j in {1,...,3}{
				\node[draw,circle,fill,scale=0.5] at (-9+\j*1.8, \yStartXPath+7.5+\i*0.5) (pX_\i_\j) {};
			}
		}
		\foreach \j in {1,...,3}{
			\node[] at (-9+\j*1.8, 7.5+\yStartXPath+9*0.5) (p_x_\j_label) {$p_X^\j$};
		}

		\begin{pgfonlayer}{bg}
			\foreach \i in {1,...,8}{
				\draw[] (pX_\i_1) -- (pX_\i_2);
				\draw[] (pX_\i_2) -- (pX_\i_3);
				\draw[] (pX_\i_3) -- (x_\i);
			}
			
			\foreach \i in {1,...,4}{
				\draw [orange] (pX_\i_1) -- (a_q1_0);
			}
			\foreach \i in {5,...,8}{
				\draw [blue] (pX_\i_1) -- (a_q1_1);
			}
			
			\foreach \i in {1,2,5,6}{
				\draw [orange] (pX_\i_2) -- (a_q2_0);
			}
			\foreach \i in {3,4,7,8}{
				\draw [blue] (pX_\i_2) -- (a_q2_1);
			}
			
			\foreach \i in {1,3,5,7}{
				\draw [orange] (pX_\i_3) -- (a_q3_0);
			}
			\foreach \i in {2,4,6,8}{
				\draw [blue] (pX_\i_3) -- (a_q3_1);
			}
		\end{pgfonlayer}
		\end{tikzpicture}
	}
	\caption{The graph $G^{\text{Q}}_3$ for the queue advantage.}
	\label{fig:G_3_Queue>Stack}
\end{figure}

\textbf{Queue behavior.}
Consider a queue based color refinement.
It splits the pairs within the paths $p_1,p_2,p_{\text{Q}}$ layer by layer. 
The splits of the $i$-th pair of $p_{\text{Q}}$ induce a split of the $i$-th vertices of the $X$-paths into the binary blocks of level $i$.
After $k+7$ many rounds, this leads to a split of $X$ into the blocks of level $k$.
Note that at the same time $\{p_{\text{end}}^0\}$ or $\{p_{\text{end}}^1\}$ will be able to split $X$.
Therefore, we know that $X$ will be fully discrete before any subset of $\mathcal{X}$ can be considered by the worklist.

\textbf{Stack behavior.}
Any stack based color refinement running on $G^{\text{Q}}_k$ splits one of the paths $p_1,p_2$ in the starting gadget before splitting the path $p_{\text{Q}}$. This induces the worst case cycling behavior of the construction from~\cite{tightLowerBound}. The unidirectional gadgets and depth first strategy hinder the algorithm from distinguishing anything else in the starting gadget before the rest of the graph has been distinguished. Therefore, no ``shortcut'' can be applied and a stack based color refinement on $G^{\text{Q}}_k$ has costs of at least $\Omega(m \log(n))$.

\section{Approximation Hardness}\label{sec:approx:hard}
Complementing our previous results, we now provide an approximation hardness result for computing optimal color refinement strategies. 
We begin by defining the optimal refinement worklist problem:
\begin{problem*}[Refinement Worklist Problem] Given a colored graph $(G, \pi)$, compute a minimal cost sequence of pairs of color classes $W=(C_1,X_1),\ldots,(C_t,X_t)$ such that:
\begin{enumerate}
	\item Refining with respect to $W$ results in the stable coloring $\pi^\infty$.
	\item For all prefixes $(C_1,X_1),\ldots,(C_s,X_s)$, the partial quotient graph obtained after refining~$C_i$ w.r.t.~$X_i$ for~$i=1,\ldots,s-1$ contains $C_s$ and $X_x$ (as unions of color classes).
\end{enumerate}
The cost of a sequence $W$ is the sum of the costs for refining with respect to all $(C_i, X_i) \in W$. 
\end{problem*}
\begin{figure}[t]
	\centering
	\begin{subfigure}{0.49\linewidth}
		\centering
	\scalebox{0.9}{\begin{tikzpicture}[scale=1,every node/.style={}]
		\node[draw=black,circle,inner sep=0,outer sep=0,minimum size=5.5mm] (a) at (0,0) {$a$};
		\node[draw=black,circle,inner sep=0,outer sep=0,minimum size=5.5mm] (b) at (1,0) {$b$};
		\node[draw=black,circle,inner sep=0,outer sep=0,minimum size=5.5mm] (c) at (2,0) {$c$};
		\node[draw=black,circle,inner sep=0,outer sep=0,minimum size=5.5mm] (d) at (3,0) {$d$};

		\node[inner sep=0,outer sep=0,dashed, circle, rounded corners, draw=black, thick, fit=(d)](FIt1) {};
		\node[outer sep=0,dashed, circle, rounded corners, draw=black, thick, fit=(b) (c) (d)](FIt2) {};
		\node[outer sep=0,dashed, circle, rounded corners, draw=black, thick, fit=(a) (b)](FIt3) {};

		\node[draw=none,circle] (u1) at (3,-1.79) {\footnotesize$U_1$};
		\node[draw=none,circle] (u2) at (2,-1.79) {\footnotesize$U_2$};
		\node[draw=none,circle] (u3) at (0.5,-1.79) {\footnotesize$U_3$};

	\end{tikzpicture}
	}	
	\caption{Set cover instance.}
	\end{subfigure}	
	\begin{subfigure}{0.5\linewidth}
		\centering
	\scalebox{0.9}{\begin{tikzpicture}[scale=0.66,every node/.style={}]
	\foreach \i in {1,...,5}{
		\node[draw=none,circle,scale=0.5] (a\i) at (1,-0.5+0.5*\i) {};
		\node[draw=none,circle,scale=0.5] (b\i) at (3,-0.5+0.5*\i) {};
		\node[draw=none,circle,scale=0.5] (c\i) at (5,-0.5+0.5*\i) {};
		\node[draw=none,circle,scale=0.5] (d\i) at (7,-0.5+0.5*\i) {};
	}
	
	\node[draw,circle,fill,scale=0.5] (U1) at (0,-2) {};
	\node[draw=none] (U1_label) at (0,-2-0.65) {\footnotesize$U_1$};

	\node[draw,circle,fill,scale=0.5] (U2) at (4,-2) {};
	\node[draw=none] (U2_label) at (4,-2-0.65) {\footnotesize$U_2$};

	\node[draw,circle,fill,scale=0.5] (U3) at (8,-2) {};
	\node[draw=none] (U3_label) at (8,-2-0.65) {\footnotesize$U_3$};

	\node[draw=none] (U1_label) at (8,1) {\footnotesize$X$};

	\foreach \i in {2,...,5}{
		\foreach \j in {1,...,3}{
		\path [ultra thin,black!50,-] (U\j) edge node[left] {} (a\i);
		\path [ultra thin,black!50,-] (U\j) edge node[left] {} (b\i);
		\path [ultra thin,black!50,-] (U\j) edge node[left] {} (c\i);
		\path [ultra thin,black!50,-] (U\j) edge node[left] {} (d\i);
		}
	}

	\node[dashed, rounded corners, draw=black, thick, fit=(a1) (b1) (c1) (d1) (a2) (b2) (c2) (d2) (a3) (b3) (c3) (d3) (a4) (b4) (c4) (d4) (a5) (b5) (c5) (d5)](FIt1) {};

	\path [thick,orange,-] (U3) edge node[left] {} (c1);
	\path [thick,orange,-] (U3) edge node[left] {} (d1);
	\path [thick,orange,-] (U2) edge node[left] {} (a1);
	\path [thick,orange,-] (U1) edge node[left] {} (a1);
	\path [thick,orange,-] (U1) edge node[left] {} (b1);
	\path [thick,orange,-] (U1) edge node[left] {} (c1);

	\foreach \i in {1,...,5}{
		\node[draw,circle,fill,scale=0.5] (a\i) at (1,-0.5+0.5*\i) {};
		\node[draw,circle,fill,scale=0.5] (b\i) at (3,-0.5+0.5*\i) {};
		\node[draw,circle,fill,scale=0.5] (c\i) at (5,-0.5+0.5*\i) {};
		\node[draw,circle,fill,scale=0.5] (d\i) at (7,-0.5+0.5*\i) {};
	}

	\end{tikzpicture}}
	\caption{Result of the reduction.}
	\end{subfigure}
\caption{Reduction of the set cover instance $S = \{a, b, c, d\}$ and $\mathcal{U} = \{\{d\}, \{b, c, d\}, \{a, b\}\}$. Orange lines indicate connections to elements of $S$, all other edges are connections to dummy elements.} \label{fig:worklistsetcover}
\end{figure}
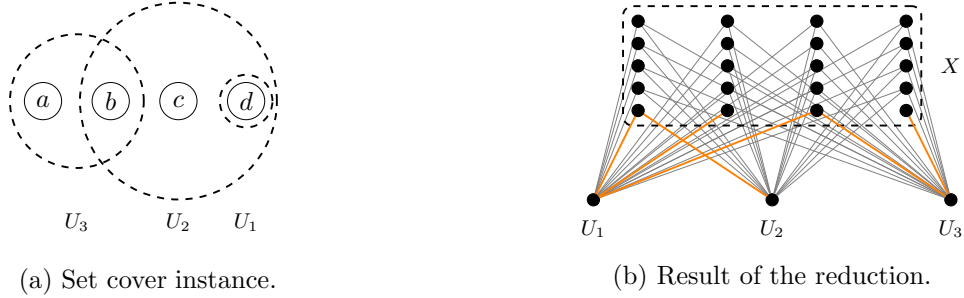
The approximation hardness result is based on a reduction from the set cover problem.
The set cover problem takes a finite universe $S$ and a set of subsets of $S$, i.e., $\mathcal{U} \subseteq 2^S$. 
The decision variant then asks whether there exists a selection of $k$ subsets in $\mathcal{U}$ whose union equals $S$. 
For simplicity, we assume $\bigcup_{U \in \mathcal{U}} U = S$.
Set cover is well-known to be $\NP$-complete. 

The optimization variant requires a minimal selection of subsets that cover $S$, i.e., a solution that minimizes $k$. 
This problem is known to be $\NP$-hard. 
More specifically, it is known that unless $\P = \NP$, polynomial-time algorithm can only reach an approximation factor of $\Omega(\log(n))$ \cite{DBLP:conf/stoc/RazS97}. 
\begin{theorem}\label{thm:approx:hard}
	Unless $\P = \NP$, polynomial-time algorithms may only reach an approximation factor of $\Omega(\log(n))$ for the optimal refinement problem.
\end{theorem}
\begin{proof}
	We reduce the optimization variant of the set cover problem to the refinement worklist problem. 
	More specifically, we reduce it in a manner which allows control of the parameters, so that the approximation hardness result of set cover immediately transfers to refinement worklists. The reduction is illustrated in Figure~\ref{fig:worklistsetcover}.
	
	Given a set cover instance $(S, \mathcal{U})$ we define a related colored graph $(G, \pi)$. 
	We create one large color class $X$ containing all elements of the universe $S$, as well as $n^2$ \emph{dummy elements} (where $n$ is the size of the set cover instance).
	Hence, the size of $X$ is $n^2 + |S|$.
	
	We add a singleton color class for each subset $U \in \mathcal{U}$, i.e., we add vertex $U$ with color $U$. 
	We connect the vertex $U$ with all vertices of $X$ \emph{except} for the elements that are contained in $U$.
	Formally, we define the edges $E(G) := \{\{U, x\} \; | x \in X \wedge x \notin U \; \}$.
	Note that $U$ has $n^2 + |S| - |U|$ connections to $X$.
	
	In the constructed graph, all elements of the universe are eventually distinguished from the dummy elements in $X$. 
	Refining $X$ with respect to $X$ is not productive, since there are no edges present and no splits occur. 
	The only way to distinguish elements of $X$ is to refine $X$ with respect to an element of $\mathcal{U}$.
	Doing so always distinguishes all the elements contained in $U \in \mathcal{U}$ from the dummy elements and other remaining elements of $X$.
	Overall, we need to refine $X$ with a subset of $\mathcal{U}$ that forms a set cover of $S$.
	
	After that, assuming all elements of $S$ have been distinguished from the dummy elements, it might be possible to split the resulting classes further through their connections to $\mathcal{U}$. 
	However, the total cost for these further refinements is bounded by $c \cdot n^2$ for some fixed constant $c$.
	
	The cost for refining $X$ with respect to $U$ is $n^2 + m$, where $m$ is the number of remaining elements of $S$ in $X$ \emph{after} the elements of $U$ have been removed.
	Since we need to choose at most $|S|$ subsets in a reasonable solution (otherwise we could remove redundant elements from the solution), and each time $X$ gets smaller by at least one element, the cost incurred by $m$ over all subsets is at most $|S|^2 \leq n^2$.
	Ignoring the cost of $m$, we get that each subset incurs additional cost of $n^2$ through the dummy elements. 
	
	Hence, the final cost is upper bounded by $c \cdot n^2 + (N_U+1) \cdot n^2 = (N_U + c + 1)n^2$ and lower bounded by $N_U n^2$, where $N_U$ is the number of chosen subsets.

	We finish our arguments with a proof by contradiction. 
	Assume there is a polynomial-time algorithm with an approximation factor in $o(\log(n))$. 
	Given a set cover instance $(S, \mathcal{U})$, we apply the polynomial-time reduction stated above. 
	Assume now we get an approximate solution with cost $x \cdot n^2$.
	We know that this implies a set cover solution with cost at most $x$.

	The optimal set cover solution with cost $x'$ would imply a worklist solution with cost at most $(x' + c)n^2$ (for a fixed $c$). 
	Hence, we know that the worklist solution also approximates the optimal solution of the original set cover instance with a factor in $o(\log(n))$.

	The set cover instance has a size in the 3rd root of the size of the refinement worklist problem.
	But since $o(\log(n)) = o(\log(\sqrt[3]{n}))$, we get a contradiction to the approximation hardness result of set cover.
	\end{proof}

\newpage
\bibliographystyle{plainurl}
\bibliography{main.bib}

\appendix

\newpage

{\noindent \Huge{\textbf{Appendix}}}\\

\noindent The appendix contains missing proofs and details for some constructions that were not presented in the main part of the paper.

\section{Details for Subsection \ref{subsec:slowconcealer} (Slow Concealer Graphs)}
\label{subsec:slowconcealer_appendix}

We now argue the correctness of the concealer graphs in more detail.
Let us first discuss some general behavior of color refinement on concealer graphs.
There are two core properties that hold for every color refinement algorithm.
The first one is that in each level~$l$, we split $X$ completely into the blocks of this level, $X_0^l,...,X_{2^l}^l$.
The other layers can only be split by the blocks of $X$, so we know that their partitions are always coarser than the one of $X$.
The second property is that out-vertices of the level~$l$ concealer gadget have to be distinguished to partition $X$ into the blocks of level~$l+1$.
Thus, also the correct input pair in this gadget has to be split.

For a coloring~$\alpha$ we denote by~$\pi_{\alpha}$ the partition induced by the coloring. The notation~$\pi_{\alpha}[X]$ indicates the restriction of the partition to a set~$X$ and we use~$\pi_{\alpha} \preceq \pi_{\alpha'}$ to indicate that the former partition is at least as fine as the latter.
Abusing notation we compare partitions of the layers of the graphs, as the are related by direct connections.

\begin{lemma}
\label{lemma:n_A}

For any coloring $\alpha_i$ with $i \geq 0$ there is a number $n_A$, such that
\begin{itemize}
	\item for all $j \leq n_A$: $\alpha_i(a_0^j) \neq \alpha_i(a_1^j)$, and
	\item for all $j > n_A$: $\alpha_i(a_0^j) = \alpha_i(a_1^j)$.
\end{itemize}

\noindent
In addition, we have either
\begin{itemize}
	\item $\pi_{\alpha_i}[X] = \{ X_q^{n_A+1} \bigm| 0 \leq q \leq 2^{n_A+1}-1 \}$ or
	\item $\pi_{\alpha_i}[X] = \{ X_q^{n_A} \bigm| 0 \leq q \leq 2^{n_A}-1 \}$.
\end{itemize}
Furthermore, it holds that $\pi_{\alpha_i}[X] \preceq \pi_{\alpha_i}[\mathcal{X}], \pi_{\alpha_i}[\mathcal{Y}], \pi_{\alpha_i}[Y]$.

\end{lemma}

\begin{proof}

For $i=0$ this is obviously true with $n_A = 0$.

For later iterations, we consider the possible splits. The start vertices $v_1$, $v_2$ and $v_3$ are not able split $X$ any further after the first iteration. Any part of $X$ in $\pi_{\alpha_i}$ can split $\mathcal{X}$, but not $A$, since $a_0^j$ and $a_1^j$ are equally connected to all $X_q^{n_A}$ and $X_q^{n_A-1}$ for all $j > n_A$.

Due to the simple one-to-$k$ connection from $X$ to $\mathcal{X}$ and because $X$ is already finer, none of these splits makes $\mathcal{X}$ finer than $X$. Since $\mathcal{X}$ is coarser than $X$, $\pi_{\alpha_i}[X]$ will not be changed by subsets of $\mathcal{X}$. With the same argument, $\mathcal{X}$ does not make $\mathcal{Y}$ coarser than $X$ and vice versa. The same holds for $\mathcal{Y}$ and $Y$.

Now consider the case that $\pi_{\alpha_i}[X] = \{ X_q^{n_A} \bigm| 0 \leq q \leq 2^{n_A}-1 \}$. Since $Y$ is coarser, no part of $Y$ can yield to an activated gadget $C_{n_A}$. Also, any other splits of concealer gadgets (which do not split input pairs) do not change the claimed property. Only distinguished out-vertices of the concealer gadgets can cause splits of $X$, which are all the $a_0^j$, $a_1^j$ with $j \leq n_A$. Out of those, only $a^{n_A}_0$ and $a^{n_A}_1$ can further split $X$ into the blocks of level $n_A+1$, not changing the claimed property. The other levels can only split $X$ into blocks of lower level, which has already been done.

Otherwise, i.e., if $\pi_{\alpha_i}[X] = \{ X_q^{n_A+1} \bigm| 0 \leq q \leq 2^{n_A+1}-1 \}$, there can be a part $Y'$ of $Y$ in $\pi_{\alpha_i}[Y]$ with $Y' = Y_q^{n_A+1}$ for some $q$. Then $Y'$ can activate $C_{n_A+1}$, which leads to a split of $a^{n_A+1}_0$ from $a^{n_A+1}_1$. This will increase $n_A$ by $1$ (since $X$ was partitioned into blocks of level $n_A+1$, the claimed property is preserved). A split of $X$ from $a^{n_A}_0$ or $a^{n_A}_1$ cannot happen in this case.

A split from the in-vertices of the concealer gadgets to $Y$ will also never make $Y$ finer than $X$, since $C_{n_A}$, the gadget which can split $Y$ into the finest partitions under all the pairs in $A$, is connected to the $n_A$-th level of $Y$, and $X$ has already been split into the blocks of level at least $n_A$.
\end{proof}

With this lemma, we can define an adversary that constructs a graph such that a specific strategy shows worst case behavior.
It should be noted that the lemma can also be stated with $a_0^j$/$a_1^j$ exchanged by $b_{i_j}^j$/$b_{i_j+1}^j$, since the split of the former is directly dependent on the split of the latter. The lemma is stated in this way so that we can reuse it at a later point.

Let $\mathcal{A}$ be the corresponding color refinement to some strategy $W$.
We construct an infinite family of graphs on which $\mathcal{A}$ has costs of $\Omega(m \log(n))$. In the family there is for each $k \in \mathbb{N}$ a graph $G_k \in \mathcal{G}_k$.
We start with a concealer graph $G \in \mathcal{G}_k$ and then successively specify the position of the correct pairs. 

Let $\pi_t$ be the partitions of $G$ that $\mathcal{A}$ produces in step $t$.
Consider an arbitrary step $t$ in the execution of $\mathcal{A}$, where $\pi_t[X] = \{ X_q^{n_A+1} \bigm| 0 \leq q \leq 2^{n_A+1}-1 \}$ for the unique $n_A$ from the previous lemma, but the correct pair $b_{i_{n_A}}^{n_A},b_{i_{n_A}+1}^{n_A}$ of the level~$n_A$ concealer gadget has not been distinguished.
We know that $\mathcal{A}$ needs to split $b_{i_{n_A}}^{n_A},b_{i_{n_A}+1}^{n_A}$ to continue to the next level.
Let $t_{\text{next}}$ be the largest $t'$ such that $\pi_{t'}[\{b_0^{n_A},...,b_{2^{n_A}-1}^{n_A}\}]=\pi_t[\{b_0^{n_A},...,b_{2^{n_A}-1}^{n_A}\}]$, i.e., the next point in time where the in-vertices of the current concealer gadget are split.
We assume w.l.o.g.\ that this is a split of an input pair.
Let $b_{i_{\text{next}}}^{n_A},b_{i_{\text{next}}+1}^{n_A}$ be the in-vertex pair which is distinguished at $t_{\text{next}}$.
An adversary can choose the concealer gadget of level~$n_A$ such that $b_{i_{\text{next}}}^{n_A},b_{i_{\text{next}}+1}^{n_A}$ is a dead end pair.

Due to Lemma~\ref{lem:concealer} and Lemma~\ref{lem:concealer:refine}, the behavior of $\mathcal{A}$ until step $t_{\text{next}}$ stays the same, no matter what concealer gadget is used in $G$.
With the previous lemma, we know that after distinguishing $b_{i_{\text{next}}}^{n_A},b_{i_{\text{next}}+1}^{n_A}$, we are in the same situation as before, i.e.\ $X$ is partitioned into the blocks of level~$n_A$, but the correct pair $b_{i_{n_A}}^{n_A},b_{i_{n_A}+1}^{n_A}$ has not been split (thus, another split of an input pair is needed to increase $n_A$).

If at step $t$ we are in the case that $\pi_t[X] = \{ X_q^{n_A} \bigm| 0 \leq q \leq 2^{n_A}-1 \}$, no important splits happen, since for all levels $l \leq n_A$, every input pair of $C_l$ has already been split and for all $l \geq n_A$, no splits of input pairs are possible.

Thus, we can repeat this changes as often as necessary such that the correct pair is split after all dead end pairs.
By doing this for each level, we get a graph on which $\mathcal{A}$ has cost $\Omega(m \log(n))$.

\section{Details for Section \ref{sec:Stack>Queue} (Stack Advantage over Queue)}
\label{sec:Stack>Queue_appendix}

We now present the graph class on which the stack worklist can be faster than the queue worklist in more detail.
We show that on this class a stack based algorithm can have linear runtime, whereas the queue version always needs time $\Omega(m \log(n))$. Furthermore, a stack that always continues with one of the smallest new classes is also always fast on this graphs.

First, we formally define a class of graphs $G^{\text{S}}_k$ ($k \in \mathbb{N}$) and then show that a color refinement using a stack as worklist may only need linear time, as well as the lower bound for any of the queue-based color refinements.

For each $k \in \mathbb{N}$, the vertex set $V(G^{\text{S}}_k)$ consists of four layers 
$$X = \{ x_i \bigm| 1 \leq i \leq 2^k \}, Y = \{ y_i \bigm| 1 \leq i \leq 2^k \},$$
$$\mathcal{X} = \{ x_i^j \bigm| 0 \leq i \leq 2^k-1, 1 \leq j \leq k \}, \mathcal{Y} = \{ y_i^j \bigm| 0 \leq i \leq 2^k-1, 1 \leq j \leq k \},$$
a connection layer
$$A = \{ a_{i,j} \bigm| 1 \leq i \leq k-1, 0 \leq j \leq 1 \}$$
and a starting gadget $v_1,v_2,v_3$.
$E(G^{\text{S}}_k)$ consists of the following sets:
\begin{itemize}
	\item $\{ (v_1,v_2) \}$, $\{(v_2,x),(v_3,x') \bigm| x \in X_0^0, x' \in X_1^0 \}$,
	\item $\{ (x_i,x_i^j),(y_i,y_i^j) \bigm| 0 \leq i \leq 2^k-1, 1 \leq j \leq k \}$,
	\item $\{ (x_i^j,y_i^{j'}) \bigm| 0 \leq i \leq 2^k-1, 1 \leq j,j' \leq k \}$,
	\item $\{ (a_{i,0},y),(a_{i,1},y') \bigm| 1 \leq i \leq k-1, y \in Y_{2j}^i, Y' \in Y_{2j+1}^i, 0 \leq j \leq 2^{i-1}-1 \}$,
	\item $\{ (a_{i,0},x),(a_{i,1},x') \bigm| 1 \leq i \leq k-2, x \in X_{2j}^{i+1}, x' \in X_{2j+1}^{i+1}, 0 \leq j \leq 2^{i+1-1}-1 \}$.
\end{itemize}

\noindent
Note that the graph has a size of $\mathcal{O}(2^k \cdot k^2)$ with $\mathcal{O}(2^k \cdot k)$ vertices and $\mathcal{O}(2^k \cdot k^2)$ edges.

$G^{\text{S}}_3$ is shown as an example in Figure~\ref{fig:G_3_Stack>Queue}.
The layout illustrates the layered structure of the different vertex sets.
The colors show the splitting possibilities of the pairs in $A$ and the blocks of $Y$.
We refer to the pair $a_{l,0},a_{l,1}$ as the $l$-th level of $A$.
The starting gadget can split $X$ into the blocks of level $1$.
The $l$-th level vertices of $A$ can split $X$ into the blocks of level $l+1$ if it was already split into blocks of level $l$ (otherwise it would be a coarser partition of $X$).
Distinguishing two $Y$-blocks of level $l$ within one block of level $l-1$ leads to distinguishing the $A$-vertices of level $l$.
Therefore, we get the same overall scheme as in Section~\ref{sec:competitiveRatio} which is formalized with the following lemma.

\begin{lemma}
	\label{lemma:XgetsDiscrete_Stack>Queue}
	
	The coarsest stable coloring $\alpha_\infty$ distinguishes all blocks of level $k$ of $X$, i.e.
	
	$$X_0^k,...,X_{2^k-1}^k \in \pi_{\alpha_\infty}.$$
	
\end{lemma}

\noindent
Analogously to Section~\ref{subsec:slowconcealer_appendix}, also in $G^{\text{S}}_k$ the pair $a_0^l,a_1^l$ has to be distinguished for $X$ to be split into the blocks of level $l+1$. Hence, we can apply Lemma~\ref{lemma:n_A} here to, which leads to the following corollary:

\begin{corollary}
	
	Apart from the first split into $X_0^1$ and $X_1^1$, which is done by $v_2$ or $v_3$, $X$ is always split by subsets of $A$. More precisely, the class which splits $X$ is either $\{a_{n_A,0}\}$ or $\{a_{n_A,1}\}$, with $n_A$ being the unique number from the previous lemma.
	
	\label{cor:AsplitsX}
\end{corollary}

\newif\ifcol
\newif\ifcoltwo
\coltrue
\coltwotrue
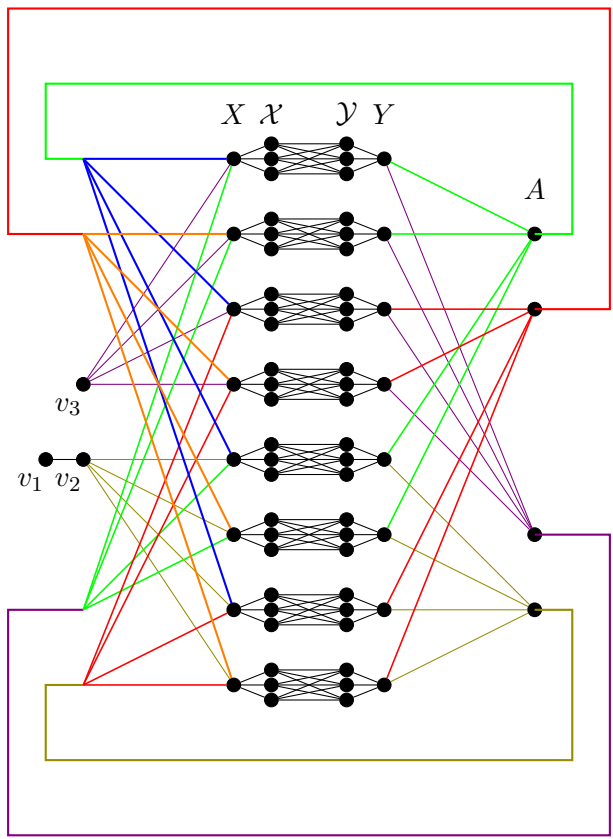
\begin{figure}
	\centering
	\begin{tikzpicture}
	\node[draw,circle,fill,scale=0.5] at (-3, 5) (v_3) {};
	\node[] at (-3.2, 4.7) (v_3_label) {$v_3$};
	\node[draw,circle,fill,scale=0.5] at (-3, 4) (v_2) {};
	\node[] at (-3.2, 3.7) (v_2_label) {$v_2$};
	\node[draw,circle,fill,scale=0.5] at (-3.5, 4) (v_1) {};
	\node[] at (-3.7, 3.7) (v_1_label) {$v_1$};
	\draw[] (v_1) -- (v_2);
	
	\node[] at (-1, 8.6) (X_label) {$X$};
	\node[] at (-0.5, 8.6) (XX_label) {$\mathcal{X}$};
	\node[] at (0.5, 8.6) (YY_label) {$\mathcal{Y}$};
	\node[] at (1, 8.6) (Y_label) {$Y$};
	\node[] at (3, 7.6) (A_label) {$A$};
	
	\node[draw,circle,fill,scale=0.5] at (3, 2) (a_1_1) {};
	\node[draw,circle,fill,scale=0.5] at (3, 3) (a_1_2) {};
	\draw \ifcoltwo [olive,thick] \else [] \fi (3,2) -- (3.5,2) -- (3.5,0) -- (-3.5,0) -- (-3.5,1) -- (-3,1);
	\draw \ifcoltwo [violet,thick] \else [] \fi (3,3) -- (4,3) -- (4,-1) -- (-4,-1) -- (-4,2) -- (-3,2);
	\node[draw,circle,fill,scale=0.5] at (3, 6) (a_2_1) {};
	\node[draw,circle,fill,scale=0.5] at (3, 7) (a_2_2) {};
	\draw \ifcoltwo [red,thick] \else [] \fi (3,6) -- (4,6) -- (4,10) -- (-4,10) -- (-4,7) -- (-3,7);
	\draw \ifcoltwo [green,thick] \else [] \fi (3,7) -- (3.5,7) -- (3.5,9) -- (-3.5,9) -- (-3.5,8) -- (-3,8);
	
	\foreach \i in {1,...,8}{
		\node[draw,circle,fill,scale=0.5] at (-1, \i) (x_\i) {};
		\node[draw,circle,fill,scale=0.5] at (1, \i) (y_\i) {};
		\foreach \j in {-1,...,1}{
			\node[draw,circle,fill,scale=0.5] at (-0.5, \i+\j*0.2) (x_\i_\j) {};
			\node[draw,circle,fill,scale=0.5] at (0.5, \i+\j*0.2) (y_\i_\j) {};
			\draw[] (x_\i) -- (x_\i_\j);
			\draw[] (y_\i) -- (y_\i_\j);
		}
		\foreach \j in {-1,...,1}{
			\foreach \k in {-1,...,1}{
				\draw[] (x_\i_\k) -- (y_\i_\j);
			}
		}
	}
	\foreach \i in {1,...,4}{
		\draw \ifcol [olive] \else [] \fi (x_\i) -- (v_2);
	}
	\foreach \i in {5,...,8}{
		\draw \ifcol [violet] \else [] \fi (x_\i) -- (v_3);
	}
	
	\draw \ifcol [olive] \else [] \fi (y_1) -- (a_1_1);
	\draw \ifcol [olive] \else [] \fi (y_2) -- (a_1_1);
	\draw \ifcol [olive] \else [] \fi (y_3) -- (a_1_1);
	\draw \ifcol [olive] \else [] \fi (y_4) -- (a_1_1);
	\draw  \ifcol [violet] \else [] \fi (y_5) -- (a_1_2);
	\draw  \ifcol [violet] \else [] \fi (y_6) -- (a_1_2);
	\draw  \ifcol [violet] \else [] \fi (y_7) -- (a_1_2);
	\draw  \ifcol [violet] \else [] \fi (y_8) -- (a_1_2);
	
	\draw \ifcol [red,semithick] \else [] \fi (y_1) -- (a_2_1);
	\draw \ifcol [red,semithick] \else [] \fi (y_2) -- (a_2_1);
	\draw \ifcol [green,semithick] \else [] \fi (y_3) -- (a_2_2);
	\draw \ifcol [green,semithick] \else [] \fi (y_4) -- (a_2_2);
	\draw \ifcol [red,semithick] \else [] \fi (y_5) -- (a_2_1);
	\draw \ifcol [red,semithick] \else [] \fi (y_6) -- (a_2_1);
	\draw \ifcol [green,semithick] \else [] \fi (y_7) -- (a_2_2);
	\draw \ifcol [green,semithick] \else [] \fi (y_8) -- (a_2_2);
	
	\draw \ifcol [red,semithick] \else [] \fi (x_1) -- (-3,1);
	\draw \ifcol [red,semithick] \else [] \fi (x_2) -- (-3,1);
	\draw \ifcol [green,semithick] \else [] \fi (x_3) -- (-3,2);
	\draw \ifcol [green,semithick] \else [] \fi (x_4) -- (-3,2);
	\draw \ifcol [red,semithick] \else [] \fi (x_5) -- (-3,1);
	\draw \ifcol [red,semithick] \else [] \fi (x_6) -- (-3,1);
	\draw \ifcol [green,semithick] \else [] \fi (x_7) -- (-3,2);
	\draw \ifcol [green,semithick] \else [] \fi (x_8) -- (-3,2);
	
	\draw \ifcol [orange,thick] \else [] \fi (x_1) -- (-3,7);
	\draw \ifcol [blue,thick] \else [] \fi (x_2) -- (-3,8);
	\draw \ifcol [orange,thick] \else [] \fi (x_3) -- (-3,7);
	\draw \ifcol [blue,thick] \else [] \fi (x_4) -- (-3,8);
	\draw \ifcol [orange,thick] \else [] \fi (x_5) -- (-3,7);
	\draw \ifcol [blue,thick] \else [] \fi (x_6) -- (-3,8);
	\draw \ifcol [orange,thick] \else [] \fi (x_7) -- (-3,7);
	\draw \ifcol [blue,thick] \else [] \fi (x_8) -- (-3,8);
	\end{tikzpicture}
	\caption[]{The graph $G^{\text{S}}_3$. The vertices $x_0...x_7$, $x_0^1...x_7^3$, $y_0^1...y_7^3$ and $y_0...y_7$ are ordered vertically from top to bottom.
		The edges incident to $A$ have been colored according to their level for an easier visual separation, those from $A$ to $X$ have been bundled for a better overview.
		For example, a separation of the level~1 blocks of $Y$ leads to a split on the first level of $A$ (yellow and purple edges on the right), which causes a split of $X$ into the blocks of level~2 (green and red edges on the left).}
	\label{fig:G_3_Stack>Queue}
\end{figure}

It should be easy to see why a queue worklist shows worst case behavior in these graphs.
The core argument needed is that for each level~$l$, right after handling the $\Theta(2^l)$ level~$l$ blocks of $X$, all $\Theta(2^l)$ blocks of $\mathcal{X}$ are handled (due to the breadth first behavior) and therefore costs of $2^k \cdot k^2$ are needed before the $l$-th level of $A$ is split.

\subsection*{Stack Behavior}

After a refinement in a stack based color refinement, the new color classes can be pushed to the stack in any order.
This order determines which class is handled next.
Therefore, it heavily influences the overall behavior of the algorithm.
We will now show that by choosing one of the smallest new classes, the linear running time can be achieved.
It should be noted that such a worklist can be implemented in the desired running time.
Furthermore, we also assume that for each split the largest new class is dropped.

We consider what can be exploited by a stack based algorithm: to distinguish $\{a_{l,0}\}$ from $\{a_{l,1}\}$, it is not necessary to split $Y$ into all blocks of level $l$.
It is actually enough to distinguish a single level~$l$ block from the rest of $Y$ to achieve that split.
Since a stack based algorithm does not handle all the blocks of $X$, $\mathcal{X}$, $\mathcal{Y}$ and $Y$ layer by layer, as a queue based algorithm does, it is possible to only handle one block of each layer per level.
In particular, for level~$l$ it is enough to do the following refinements: one w.r.t.\ a level~$l$ block of $X$, then one w.r.t.\ a single block of $\mathcal{X}$ and one w.r.t.\ a block of $\mathcal{Y}$.
At last, a single refinement w.r.t.\ a level~$l$ block of $Y$ splits the $l$-th level of $A$ and we can continue with a level~$l+1$ block of $X$.
Hence, the running time for level $l$ is $\mathcal{O}(2^{k-l} \cdot k^2)$, which then gives a total running time of $\mathcal{O}(2^{k} \cdot k^2)$.
The considered stack does exactly this.
When propagating through the layers, the only classes to choose for continuation are blocks of the same level, as the larger ones are always dropped.
After the refinement w.r.t.\ $a_0^l$ or $a_1^l$, the new classes are level~$l$ blocks of $Y$ and level~$l+1$ blocks of $X$.
Since the level~$l+1$ blocks are smaller, one of them is chosen as the next set to refine with.
Hence, we get the desired behavior.

\section{Details for Section \ref{sec:Queue>Stack} (Queue Advantage over Stack)}
\label{sec:Queue>Stack_appendix}
We now present the graph class on which the queue worklist can be faster than the stack worklist in more detail. We start by proving the claimed properties of the starting gadget. Then, we give a full definition of the graph class that was shown in Figure~\ref{fig:G_3_Queue>Stack}. At last, we describe our arguments on the queue and stack behavior in more detail.

\subsection*{Starting gadget}
To construct the gadget shown in Figure~\ref{fig:startingGadget_Queue>Stack}, we add the following vertices and edges to an $\AND_2$ graph:
\begin{itemize}
	\item four starting vertices $s_1,...,s_4$, connected to the in-vertices $b_0,...,b_3$ by the edges\\
	$\{ (s_i,b_{i-1}) \bigm| 0 \leq i \leq 3 \}$
	\item two paths of pairs $p_1^{1,1},p_1^{2,1},p_1^{1,2},p_1^{2,2},p_1^{1,3},p_1^{2,3}$ and $p_2^{1,1},p_2^{2,1},p_2^{1,2},p_2^{2,2},p_2^{1,3},p_2^{2,3}$
	with the edges $\{ (p_i^{i',j},p_i^{i',j+1}) \bigm| i,i'\in\{1,2\}, j \in \{1,2\} \}$
	\item edges between $B$ and the two new paths:\\
	$\{ (b_0,p_1^{1,1}),(b_1,p_1^{2,1}),(b_2,p_2^{1,1}),(b_3,p_2^{2,1}) \}$
\end{itemize}

\noindent
To understand the functionality, we consider the behavior of Algorithm~\ref{alg:refine} on the gadget. Assume we have a coloring $\alpha$ with
\begin{itemize}
	\item $\alpha(b_i)=\alpha(b_j),\alpha(c_i)=\alpha(c_j)$ for all $i,j \in \{0,...,3\}$,
	\item $\alpha(a_0)=\alpha(a_1)$,
	\item $\alpha(p_1^{1,j})=\alpha(p_1^{2,j}),\alpha(p_2^{1,j})=\alpha(p_2^{2,j})$ for $j\in\{1,2,3\}$ and
	\item $\{s_1\},...,\{s_4\} \in \pi_\alpha$.
\end{itemize}
Any color refinement will eventually distinguish the pairs within the paths as well as the two out-vertices of the $\AND_2$ gadget, but a stack based one does it in a specific order.

We now consider a stack based algorithm that starts on this coloring. The intuitive idea is that any stack based color refinement will handle one of the paths before the vertices $a_0$ or $a_1$. It will first distinguish $p_1^{1,3}$ from $p_1^{2,3}$ or $p_2^{1,3}$ from $p_2^{2,3}$ before distinguishing $a_0$ from $a_1$, and, more importantly, it handles one of those new singleton classes $\{p_1^{1,3}\}$,$\{p_1^{2,3}\}$,$\{p_2^{1,3}\}$,$\{p_2^{2,3}\}$ before handling $\{a_0\}$ or $\{a_1$\}.

To see that, consider which splits are possible. The only classes for which a refinement will lead to splits of other classes are the singleton classes of the starting vertices. Each of them will individualize its partner in-vertex. So let us assume, w.l.o.g., that we start with $\{s_1\}$. This gives us two new classes $\{b_0\}$ and $\{b_1,b_2,b_3\}$, of which the larger one is dropped. So we continue by refining w.r.t.\ $\{b_0\}$. This leads to a split of the middle layer and also distinguishes $p_1^{1,1}$ and $p_1^{2,1}$, thus the first claimed property holds. Even if a new subclass of the middle layer is handled first, we know that it will not yield a split of $a_0$ and $a_1$, for which it is necessary to also distinguish $b_2$ and $b_3$. After the middle layer is handled, due to the stack's depth first behavior, $\{p_1^{1,1}\}$ or $\{p_1^{2,1}\}$ are dequeued, before any of the other starting vertices. Then the whole path $p_1$ is made discrete and handled step by step. Thus the second property also holds. By appending other graphs at the end of the paths, we can enforce certain splits to happen before $a_0$ and $a_1$ are distinguished and handled. 

For queue based algorithms, this is not the case. Due to the breadth first behavior, all starting vertices are directly handled and the $\AND$ gadget is activated. Thus, all induced splits happen in a layered fashion, where $a_0$ and $a_1$ are in the same layer as the vertices $p_1^{1,2}$, $p_1^{2,2}$, $p_2^{1,2}$ and $p_2^{2,2}$. They are therefore handled before the end vertices of the paths.

\subsection*{Definition of the Graph Class}

We now define the graph $G^{\text{Q}}_k$ for each $k \in \mathbb{N}$. At first, we simply list all vertices and edges, as the nomenclature is important for the arguments on the behavior of color refinement. We start with the vertex set, which consists of the following subsets:
\begin{itemize}
	\item The sets $X$ and $Y$: $ X = \{ x_i \bigm| i \in \mathcal{B}_k \}, Y = \{ y_i \bigm| i \in \mathcal{B}_k \} $
	\item The sets $\mathcal{X}$ and $\mathcal{Y}$: $ \mathcal{X} = \{ x_i^j \bigm| i \in \mathcal{B}_k, 1 \leq j \leq k \}, \mathcal{Y} = \{ y_i^j \bigm| i \in \mathcal{B}_k, 1 \leq j \leq k \} $
	\item An $\AND_l$ gadget for each $1 \leq l \leq k-1$ with in-vertices $b_l^0,...,b_l^{2^l-1}$ and out-vertices $a_l^0$ and $a_l^1$
	\item The start vertices $\{ s_1, s_2, s_3, s_4 \}$
	\item The $\AND$ gadget within the starting gadget: $\{ b_{\text{start}}^i,c_{\text{start}}^i \bigm| 0 \leq i \leq 3 \} \cup \{ a_{\text{start}}^0,a_{\text{start}}^1 \}$
	\item The $X$-paths $ \{ p_{x_i}^j \bigm| 0 \leq i \leq 2^k, 0 \leq j \leq k \} $
	\item The stack paths $ p_1 = \{ p_1^{1,i},p_1^{2,i} \bigm| 0 \leq i \leq k+2 \} $ and $ p_2 = \{ p_2^{1,i},p_2^{2,i} \bigm| 0 \leq i \leq k+2 \} $ with end vertices $\{p_{\text{end}}^1,p_{\text{end}}^2\}$
	\item The queue path $ p_{\text{Q}} = \{ p_{\text{Q}}^{i,0},p_{\text{Q}}^{i,1} \bigm| 0 \leq i \leq k \} $
	\item The unidirectional gadgets connecting the queue path with the $X$-paths:\\ $\{ b_{\text{Q}}^{l,i},c_{\text{Q}}^{l,i} \bigm| 0 \leq l \leq k, 0 \leq i \leq 3 \} \cup \{ a_{\text{Q}}^{l,0},a_{\text{Q}}^{l,1} \bigm| 0 \leq l \leq k \}$
	\item The unidirectional gadgets connecting the stack paths with their end vertices: $\{ b_{\text{end}}^{j,i},c_{\text{end}}^{j,i} \bigm| 1 \leq j \leq 2, 0 \leq i \leq 3 \} \cup \{ a_{\text{end}}^{j,0},a_{\text{end}}^{j,1} \bigm| 1 \leq j \leq 2 \}$
\end{itemize}
The graph has the following edges:
\begin{itemize}
	\item $\{ (x_i,x_i^j),(y_i,y_i^j) \bigm| i \in \mathcal{B}_k, 1 \leq j \leq k \}$
	\item $\{ (x_i^j,y_i^{j'}) \bigm| i \in \mathcal{B}_k, 1 \leq j,j' \leq k \}$
	\item The edges of the $\AND_l$ gadgets for each $1 \leq l \leq k-1$ (We will omit naming for better readability, as they are not important as long as the behavior of the gadget is known)
	\item $\{ (y,b_l^i) \bigm| y \in Y_i^l, 1 \leq l \leq k-1, 0 \leq i \leq 2^l-1 \}$
	\item $\{ (a_l^0,x),(a_l^1,x') \bigm| x \in X_{2i}^{l+1}, x' \in X_{2i+1}^{l+1}, 1 \leq l \leq k-1, 0 \leq i \leq 2^l-1 \}$
	
	\item $\{ (p_{x_i}^j,p_{x_i}^{j+1}) \bigm| 0 \leq i \leq 2^k, 0 \leq j \leq k-1 \}$
	\item $\{ (p_{x_i}^k,x_i) \bigm| 0 \leq i \leq 2^k \}$
	
	\item $\{ (s_i,b_{\text{start}}^{i-1}) \bigm| 1 \leq i \leq 4 \}$
	\item $\{ (b_{\text{start}}^0,c_{\text{start}}^0),(b_{\text{start}}^0,c_{\text{start}}^2),(b_{\text{start}}^1,c_{\text{start}}^1),(b_{\text{start}}^1,c_{\text{start}}^3)\}$ \\ $\{(b_{\text{start}}^2,c_{\text{start}}^1),(b_{\text{start}}^2,c_{\text{start}}^2),(b_{\text{start}}^3,c_{\text{start}}^0),(b_{\text{start}}^3,c_{\text{start}}^3) \}$
	\item $\{ (c_{\text{start}}^0,a_{\text{start}}^0),(c_{\text{start}}^1,a_{\text{start}}^0),(c_{\text{start}}^2,a_{\text{start}}^1),(c_{\text{start}}^3,a_{\text{start}}^1) \}$
	
	\item $\{ (b_{\text{Q}}^{l,0},c_{\text{Q}}^{l,0}),(b_{\text{Q}}^{l,0},c_{\text{Q}}^{l,2}),(b_{\text{Q}}^{l,1},c_{\text{Q}}^{l,1}),(b_{\text{Q}}^{l,1},c_{\text{Q}}^{l,3}),(b_{\text{Q}}^{l,2},c_{\text{Q}}^{l,1}),(b_{\text{Q}}^{l,2},c_{\text{Q}}^{l,2}),(b_{\text{Q}}^{l,3},c_{\text{Q}}^{l,0}),(b_{\text{Q}}^{l,3},c_{\text{start}}^{l,3}) \}$\\ for each $l$ in $\{0,...,k\}$
	\item $\{ (c_{\text{Q}}^{l,0},a_{\text{Q}}^{l,0}),(c_{\text{Q}}^{l,1},a_{\text{Q}}^{l,0}),(c_{\text{Q}}^{l,2},a_{\text{Q}}^{l,1}),(c_{\text{Q}}^{l,3},a_{\text{Q}}^{l,1}) \bigm| 0\leq l \leq k \}$
	
	\item $\{ (b_{\text{end}}^{i,0},c_{\text{end}}^{i,0}),(b_{\text{end}}^{i,0},c_{\text{end}}^{i,2}),(b_{\text{end}}^{i,1},c_{\text{end}}^{i,1}),(b_{\text{end}}^{i,1},c_{\text{end}}^{i,3}) \bigm| 1 \leq i \leq 2 \}$\\ $\{(b_{\text{end}}^{i,2},c_{\text{end}}^{i,1}),(b_{\text{end}}^{i,2},c_{\text{end}}^{i,2}),(b_{\text{end}}^{i,3},c_{\text{end}}^{i,0}),(b_{\text{end}}^{i,3},c_{\text{end}}^{i,3}) \bigm| 1 \leq i \leq 2 \}$
	\item $\{ (c_{\text{end}}^{i,0},a_{\text{end}}^{i,0}),(c_{\text{end}}^{i,1},a_{\text{end}}^{i,0}),(c_{\text{end}}^{i,2},a_{\text{end}}^{i,1}),(c_{\text{end}}^{i,3},a_{\text{end}}^{i,1}) \bigm| 1 \leq i \leq 2 \}$
	
	\item $ \{ (p_{\text{Q}}^{1,i},p_{\text{Q}}^{1,i+1}),(p_{\text{Q}}^{2,i},p_{\text{Q}}^{2,i+1}),(p_{\text{Q}}^{1,i},p_{\text{Q}}^{1,i+1}),(p_{\text{Q}}^{2,i},p_{\text{Q}}^{2,i+1}) \bigm| 0 \leq i \leq k-1 \}$
	
	\item $ \{ (p_1^{1,i},p_1^{1,i+1}),(p_1^{2,i},p_1^{2,i+1}),(p_2^{1,i},p_2^{1,i+1}),(p_2^{2,i},p_2^{2,i+1}) \bigm| 0 \leq i \leq k+1 \}$
	\item $ \{ (p_1^{1,k+2},b_{\text{end}}^{1,0}),(p_1^{1,k+2},b_{\text{end}}^{1,2}), (p_1^{2,k+2},b_{\text{end}}^{1,1}),(p_1^{2,k+2},b_{\text{end}}^{1,3})\}$\\
	$\{(p_2^{1,k+2},b_{\text{end}}^{2,0}),(p_2^{1,k+2},b_{\text{end}}^{2,2}),
	(p_2^{2,k+2},b_{\text{end}}^{2,1}),(p_2^{2,k+2},b_{\text{end}}^{2,3}) \} $
	\item $ \{ (a_{\text{end}}^{1,1},p_{\text{end}}^1),(a_{\text{end}}^{1,2},p_{\text{end}}^2),(a_{\text{end}}^{1,1},p_{\text{end}}^1),(a_{\text{end}}^{2,2},p_{\text{end}}^2) \} $
	
	\item $\{ (p_{\text{end}}^1,x),(p_{\text{end}}^2,x') \bigm| x \in X_0^1, x' \in X_1^1 \}$
	
	\item $\{ (a_{\text{Q}}^{l,1},p_{x_i}^l),(a_{\text{Q}}^{l,2},p_{x_i'}^l) \bigm| i \in \mathcal{B}_0^1, i' \in \mathcal{B}_1^1 \}$
\end{itemize}
Figure~\ref{fig:G_3_Queue>Stack} shows the example graph $G_3$. Again colors are used to get a better view on the splitting powers of certain classes. We will assume that $s_1,...,s_4$ are individualized initially. To achieve that, we can simply connect each one with a unique amount of leaf nodes. As the set of degrees in the graph is constant, this will not change the asymptotic size of the graph $G^{\text{Q}}_k$. Then it holds that after the first refinement every vertex is colored with its degree, including that the four start vertices are individualized. So it does not change anything to assume this coloring as the initial one and omit the leaves connected to the start vertices.

Starting with this coloring, it holds that all the vertex pairs on the paths will be distinguished eventually. This is easy to see by the following observations. The individualized start vertices activate the $\AND$ gadget in the starting gadget, which causes $p_{\text{Q}}$ to be split. It also leads to the pairs in the paths $p_1$ and $p_2$ to be distinguished, as those two paths are directly connected to the start vertices. Furthermore, in this graph class $X$ and $Y$ again become fully discrete and therefore the $X$-paths do so as well. This is caused by the splits on $p_1$ and $p_2$ (amongst others), since they yield the same initial split on $X$ as in the other constructions in Section~\ref{sec:Stack>Queue} and Section~\ref{sec:competitiveRatio}. This allows us to use Lemma~\ref{lemma:XgetsDiscrete_Stack>Queue} here as well to conclude that $X$ becomes discrete.

We formalize these results for the new graph class in the following lemma.

\begin{lemma}
	
	Let $\alpha$ be a coloring for $G^{\text{Q}}_k$ with $\alpha(s_i)=i$ for $1 \leq i \leq 4$ and $\alpha(v) = d(v)+4$ for all $v \in V(G^{\text{Q}}_k) \setminus \{s_1,s_2,s_3,s_4\}$. For $\alpha_\infty$ the following holds:
	\begin{itemize}
		\item $X_i^k,Y_i^k,\mathcal{X}_i^k,\mathcal{Y}_i^k \in \pi_{\alpha_\infty}$ for all $0 \leq i \leq 2^k-1$
		\item $\{p_{x_i}^{l}\} \in \pi_{\alpha_\infty}$ for all $0 \leq i \leq 2^k-1,1 \leq l \leq k$.
		\item $\alpha_\infty(p_i^{1,j}) \neq \alpha_\infty(p_i^{2,j})$ for all $1 \leq i \leq 2, 1 \leq j \leq k+2$
		\item $\alpha_\infty(p_{\text{Q}}^{1,j}) \neq \alpha_\infty(p_{\text{Q}}^{2,j})$ for all $1 \leq j \leq k$
		\item $\alpha_\infty(p_{\text{end}}^1) \neq \alpha_\infty(p_{\text{end}}^2)$ for all $1 \leq j \leq k$
	\end{itemize}
	
\end{lemma}

\noindent
Now let us consider the size of $G^{\text{Q}}_k$. The old part of the graph still has $\mathcal{O}(2^k \cdot k)$ vertices and $\mathcal{O}(2^k \cdot k^2)$ edges. Each $\AND_l$ gadget has $\mathcal{O}(2^l)$ vertices and edges. There are $2^k \cdot k$ vertices $p_{x_i}^j$ and $2^k \cdot (k-1)$ edges connecting them. $p_1$, $p_2$ and $p_{\text{Q}}$ have size $\mathcal{O}(k)$ (edges and vertices). The starting gadget has constant size as well as all the gadget connecting the paths, of which we have $\mathcal{O}(k)$ many. This gives a total amount of $\mathcal{O}(2^k \cdot k)$ vertices and $\mathcal{O}(2^k \cdot k^2)$ edges.

\subsection*{General Observations}

We start with some observations about the splitting behavior of Algorithm~\ref{alg:refine} on $G^{\text{Q}}_k$, independent of the chosen worklist.

Let us consider the three paths of the starting gadget $p_1$, $p_2$ and $p_{\text{Q}}$. Some of the vertex pairs within these paths are connected to unidirectional gadgets. From Section~\ref{sec:gadgets} we know that those pairs will never be distinguished due to splits on other vertices of the unidirectional gadgets. So we know that $p_i^{1,k+2}$ and $p_i^{2,k+2}$ are only distinguished by a refinement w.r.t.\ $p_i^{1,k+1}$ or $p_i^{2,k+1}$ for $i \in \{1,2\}$. Therefore, the same of course holds for $p_i^{1,j}$ and $p_i^{2,j}$ with $j \in \{ 2,...,k+1 \}$, since these pairs are not connected to any other vertices. So we know that all of these splits are initialized by distinguishing the corresponding pair of start vertices.

With the same argument we get that the $p_{\text{Q}}^{i,j}$ can only be distinguished by the neighboring $p_{\text{Q}}^{i,j+1}$ or $p_{\text{Q}}^{i,j-1}$, as the connected unidirectional gadget will not cause a split. With this knowledge we can conclude that splits of vertex pairs on this path are all initialized by distinguishing the out-vertices in the starting $\AND$ gadget, i.e.\ $a_{\text{start}}^0$ and $a_{\text{start}}^1$.

In conclusion, this formally means that if we remove the starting $\AND$ gadget and $s_1,...,s_4$, the corresponding paths will never be split.

\begin{lemma}
	\label{lemma:Queue>Stack_OnlyStartSplitsPaths}
	
	Let $G'^{\text{Q}}_k$ be the graph $G^{\text{Q}}_k$ without the starting $\AND$ gadget and the start vertices and $\alpha$ a coloring for $G'^{\text{Q}}_k$ with:
	\begin{itemize}
		\item $\alpha(v) = \alpha(v')$ for all $v,v' \in \{p_i^{j,l} \bigm| i,j \in \{1,2\}, 0 \leq l \leq k+2\}$
		\item $\alpha(v) = \alpha(v')$ for all $v,v' \in \{p_{\text{Q}}^{j,l'} \bigm| j \in \{1,2\}, 0 \leq l' \leq k\} $
		\item $\alpha(v) = \alpha(v')$ for all $v,v' \in \{b_{\text{end}}^{j,i},c_{\text{end}}^{j,i} \bigm| j \in \{1,2\}, i \in \{0,..,3\}\}$
		\item $\alpha(v) = \alpha(v')$ for all $v,v' \in \{b_{\text{Q}}^{j,l},c_{\text{Q}}^{j,l} \bigm| j \in \{1,2\}, l \in \{1,..,k\}\}$
		\item $\alpha(v)$ arbitrary for all other vertices $v$
	\end{itemize} 
	Then for the coarsest stable coloring $\alpha_\infty$ it holds that:
	\begin{itemize}
		\item $\alpha_\infty(p_1^{1,j}) = \alpha_\infty(p_1^{2,j})$ for all $0 \leq j \leq k+2$
		\item $\alpha_\infty(p_2^{1,j}) = \alpha_\infty(p_2^{2,j})$ for all $0 \leq j \leq k+2$
		\item $\alpha_\infty(p_1^{1,j}) = \alpha_\infty(p_1^{2,j})$ for all $0 \leq j \leq k$
	\end{itemize}
	
\end{lemma}

\subsection*{Stack Behavior}

We now discuss the behavior of a stack based algorithm on $G^{\text{Q}}_k$ and show that the running time of any stack based algorithm has a lower bound of $\Omega(m \log(n))$. We start by considering which behaviors are initiated by splits within the starting gadget. For a better overview, we assume that after the first refinement $s_1,...,s_4$ are all individualized and all other vertices are distinguished by their degree.

The first thing to consider is the following fact: $b_{\text{start}}^0$ and $b_{\text{start}}^1$ can only be distinguished when $\{s_1\}$ or $\{s_2\}$ has been handled.
The same holds for $b_{\text{start}}^2$ and $b_{\text{start}}^3$ with the condition that either $s_3$ or $s_4$ has been handled.
Formally this means that if we remove the individualized $s_1$ and $s_2$ from $G^{\text{Q}}_k$, then in the coarsest stable coloring $b_{\text{start}}^0$ and $b_{\text{start}}^1$ will be in the same color class.
We refer to this graphs as $G'^{\text{Q}}_k$.

\begin{lemma}
	\label{lemma:Queue>Stack_b0b1NotSplitInG'}
	
	Consider $G'^{\text{Q}}_k$ with a coloring $\alpha$ such that $\alpha(s_3)=1$, $\alpha(s_4)=2$ and $\alpha(v) = 2+d(v)$ for all $v \in V(G'^{\text{Q}}_k) \setminus \{s_3,s_4\}$. For the coarsest stable coloring $\alpha_\infty$ it holds that:
	$$ \alpha_\infty(b_{\text{start}}^0) = \alpha_\infty(b_{\text{start}}^1) $$
	
\end{lemma}

\begin{proof}
	
	From Lemma~\ref{lemma:Queue>Stack_OnlyStartSplitsPaths} we know that $a_{\text{start}}^0$ and $a_{\text{start}}^1$ are never distinguished by a refinement w.r.t.\ a subclass of $p_{\text{Q}}$. So they can only be split by a refinement w.r.t\ a class within the $\AND$ gadget.
	
	With the changes made, $b_{\text{start}}^0$ and $b_{\text{start}}^1$ are only connected to the path $p_1$ (which can only be split by $b_{\text{start}}^0$ or $b_{\text{start}}^1$, due to the unidirectional gadget at the other end) and the middle layer of the $\AND$ gadget. So the only possible splits on $b_{\text{start}}^0$ and $b_{\text{start}}^1$ would have to be done by refinements w.r.t\ a class within the $\AND$ gadget. Because we know that the out-vertices are not distinguished by any other refinement, such a split would contradict the known behavior of the $\AND$ gadget, as it would mean that distinguishing only one of the pairs of in-vertices could activate the gadget.
	
	Therefore, we can conclude that $b_{\text{start}}^0$ and $b_{\text{start}}^1$ will never be distinguished.
	
\end{proof}

\noindent
With this lemma and the knowledge about the behavior of the $\AND$ gadget we can conclude that $a_{\text{start}}^0$ and $a_{\text{start}}^1$ will never be distinguished, as well as the vertex pairs on $p_1$ and $p_{\text{Q}}$.

\begin{corollary}
\label{cor:Queue>Stack_PathsNotSplitUntilStartFullySplit}

Let $G'^{\text{Q}}_k$, $\alpha$ and $\alpha_\infty$ be defined like in the previous lemma. We also know about $\alpha_\infty$:
\begin{itemize}
	\item $ \alpha_\infty(a_{\text{start}}^0) = \alpha_\infty(a_{\text{start}}^2) $
	\item $ \alpha_\infty(p_{\text{Q}}^{1,i}) = \alpha_\infty(p_{\text{Q}}^{2,i}) $ for all $ 1 \leq i \leq k$
	\item $ \alpha_\infty(p_{1}^{1,i}) = \alpha_\infty(p_{1}^{2,i}) $ for all $ 1 \leq i \leq k+2$
\end{itemize}

\end{corollary}

\noindent
Now consider the execution of a stack based algorithm on the original $G^{\text{Q}}_k$.
After the initial refinement, we have the singleton sets $\{s_1\}$, $\{s_2\}$, $\{s_3\}$ and $\{s_4\}$ in the worklist as well as the classes of vertices with equal degree.

W.l.o.g.\ $\{s_4\}$ is handled first out of these four classes.
Since the other singleton sets will never be split again, they will stay at the bottom of the stack until all refinements that are induced by $\{s_4\}$ are done.
Since $s_4$ is distinguished from $s_3$, and $s_1$ and $s_2$ can be ignored, all splits that are done in the color refinement on $G'^{\text{Q}}_k$ can also be done in here.
In fact, only distinguishing $s_3$ from $s_4$ initiates the cycling behavior that we already know, which means that the refinement w.r.t.\ the class $\{s_4\}$ inevitably leads to all the splits that are done when computing the coarsest stable coloring of $G'^{\text{Q}}_k$.
Thus, a stack based color refinement on $G^{\text{Q}}_k$ first computes the coarsest stable coloring of its subgraph $G'^{\text{Q}}_k$, before handling $\{s_1\}$, $\{s_2\}$ or $\{s_3\}$, which means that its running time on $G'^{\text{Q}}_k$ is a lower bound for the running time on $G^{\text{Q}}_k$.

Now consider the costs of an arbitrary color refinement on $G'^{\text{Q}}_k$.
Corollary~\ref{cor:Queue>Stack_PathsNotSplitUntilStartFullySplit} shows us that the path $p_{\text{Q}}$ will never be spit in the computation of $\alpha_\infty^{G'^{\text{Q}}_k}$.
This implies that the levels $\{p_{x_i}^l\}$ for $1 \leq l \leq k$ are never split by refinements w.r.t.\ $\{a_{\text{Q}}^{0,l}\}$ or $\{a_{\text{Q}}^{1,l}\}$.
Thus, the splits on them are only induced by refinements w.r.t.\ subclasses of $X$ and they can therefore never become finer than $X$.

\begin{lemma}
	
	When computing the coarsest stable coloring $\alpha_\infty$ refining a coloring $\alpha$ of $G'^{\text{Q}}_k$, for each coloring $\alpha_i$ that appears during the computation we know that $\pi_{\alpha_i}[\{p_{x_i}^l\}]$ is coarser than $\pi_{\alpha_i}[X]$ for each $1 \leq l \leq k$.
	
\end{lemma}

\noindent
From here we can conclude that $X$ is split only by two other classes of vertices, like in Section~\ref{sec:competitiveRatio}: the singleton classes of the out-layer in the $\AND_l$ gadgets ($\{a_l^0\}$ or $\{a_l^1\}$ for $1 \leq l \leq k-1$), and the end vertices of the starting gadget ($\{p_{\text{end}}^0\}$ or $\{p_{\text{end}}^1\}$).
This allows us to use the same arguments as in the lower bounds paper by Berkholz et. al~\cite{tightLowerBound}, so any color refinement on $G'^{\text{Q}}_k$ has costs of at least $\Omega(m \log(n))$.
Therefore, we have a lower bound on the costs of any stack based color refinement on $G^{\text{Q}}_k$.

\subsection*{Queue Behavior}

In this section, we show the linear running time for a queue based algorithm. We start with formulating the concept of ``rounds'' mentioned before and then use it to show that the algorithm behaves as desired.

Let the depth of a class of vertices in the worklist be defined as follows: the first element in the queue ($V(G)$) has depth $0$. If a new class $C_{\text{new}}$ is pushed to the worklist while handling a class of depth $i$, we define the depth of $C_{\text{new}}$ to be $i+1$. All classes of depth $i$ are handled consecutively and before classes of depth $i+1$. We call the iterations in which all classes of depth $i$ are handled the $i$-th round of the color refinement.

For $G^{\text{Q}}_k$, the sets of depth $1$ are $\{s_1\}$, $\{s_2\}$, $\{s_3\}$, $\{s_4\}$ and the sets $\{ v \in V(G^{\text{Q}}_k) \bigm| \deg(v) = d \}$ for all appearing degrees $d$. $\mathcal{X} \cup \mathcal{Y}$ is one of those classes and is split into $\mathcal{X}$ and $\mathcal{Y}$ in round~$1$ by a refinement w.r.t.\ $X$ or $Y$. One of these sets is handled in round~$2$, the other one is dropped as the class $\mathcal{X} \cup \mathcal{Y}$ was split in halves. Afterwards they are never split and therefore never handled until $X$ is split. Let the first round in which $X$ is split be $i_X$.

Without the edges of $\mathcal{X}$ and $\mathcal{Y}$ $G^{\text{Q}}_k$ has only $\mathcal{O}(2^k \cdot k)$ other edges and vertices. So the total costs of all rounds~$i$ with $2 < i < i_X$ together are at most $\mathcal{O}(2^k \cdot k^2) = \mathcal{O}(2^k \cdot k \cdot \log(2^k \cdot k))$. Note that we just apply the known upper bound of $\mathcal{O}(m \cdot \log(n))$ here. Since we know that the first three rounds also have costs of $\mathcal{O}(2^k \cdot k^2)$, the total running time until round~$i_X$ is in $\mathcal{O}(2^k \cdot k^2)$. Next we show that after round~$i_X$ is completed, $X$ will be discrete. From there we can easily conclude that the total costs are $\mathcal{O}(2^k \cdot k^2) = \mathcal{O}(m)$.

Now consider the first rounds on $G^{\text{Q}}_k$. Round~$0$ refines w.r.t.\ $V(G^{\text{Q}}_k)$ and splits the vertices into the classes $\{s_1\}$, $\{s_2\}$, $\{s_3\}$, $\{s_4\}$ and the sets $\{ v \in V(G^{\text{Q}}_k) \bigm| \deg(v) = d \}$ for all appearing degrees $d$, as already mentioned. In round~$1$, the refinements w.r.t.\ $\{s_j\}$, $1 \leq j \leq 4$, distinguish the vertices $b_{\text{start}}^0$,...,$b_{\text{start}}^3$. These new singleton sets then distinguish $p_1^{1,1}$ from $p_1^{2,1}$ and $p_2^{1,1}$ from $p_2^{2,1}$ as well as the four vertices $c_{\text{start}}^0$,...,$c_{\text{start}}^3$ in round~$2$.

In round~$3$ $a_{\text{start}}^0$,$a_{\text{start}}^1$, $p_1^{1,2}$, $p_1^{2,2}$, $p_2^{1,2}$ and $p_2^{2,2}$ are distinguished. This leads to a sequential split of the three paths, so for $i \in \{1,...,k\}$ round~$i+3$ distinguishes $p_{\text{Q}}^{1,i}$ from $p_{\text{Q}}^{2,i}$, $p_1^{1,i+2}$ from $p_1^{2,i+2}$ and $p_2^{1,i+2}$ from $p_2^{2,i+2}$.

Also round~$i+4$ splits $b_{\text{Q}}^{i,0}$,...,$b_{\text{Q}}^{i,3}$, therefore in round~$i+5$ $c_{\text{Q}}^{i,0}$,...,$c_{\text{Q}}^{i,3}$ are split and round~$i+6$ distinguishes $a_{\text{Q}}^{i,0}$ from $a_{\text{Q}}^{i,1}$. This gives us that in round~$8$ the set $\{p_{x_j}^1 \bigm| j \in \mathcal{B}_k\}$ is split into the blocks of level~$1$ and for $i \in \{2,...,k\}$ in round~$i+7$ the level~$i-1$ blocks of $\{p_{x_j}^{i-1} \bigm| j \in \mathcal{B}_k\}$ and $a_{\text{Q}}^{i,0}$ or $a_{\text{Q}}^{i,1}$ will split $\{p_{x_j}^i \bigm| j \in \mathcal{B}_k\}$ into the blocks of level~$i$. This means that after round~$k+7$ the set $\{p_{x_j}^k \bigm| j \in \mathcal{B}_k\}$ is split into the blocks of level $k$ and therefore discrete.

Also in round~$k+3$ the vertices $p_1^{1,k+2}$ and $p_1^{2,k+2}$ are distinguished as well as $p_2^{1,k+2}$ and $p_2^{2,k+2}$. Thus round~$k+4$ splits $b_{\text{end}}^{1,0}$,...,$b_{\text{end}}^{1,3}$ and $b_{\text{end}}^{2,0}$,...,$b_{\text{end}}^{2,3}$, which means that round~$k+5$ splits $c_{\text{end}}^{1,0}$,...,$c_{\text{end}}^{1,3}$ and $c_{\text{end}}^{2,0}$,...,$c_{\text{end}}^{2,3}$. In round~$k+6$ $a_{\text{end}}^{1,0}$ is distinguished from $a_{\text{end}}^{1,1}$ and $a_{\text{end}}^{2,0}$ from $a_{\text{end}}^{2,1}$. Therefore, $p_{\text{end}}^1$ and $p_{\text{end}}^2$ are distinguished in round~$k+7$, at the same time that $\{p_{x_j}^k \bigm| j \in \mathcal{B}_k\}$ is split.

So in the next round the singleton blocks of $\{p_{x_j}^k \bigm| j \in \mathcal{B}_k\}$ as well as the sets $\{p_{\text{end}}^1\}$ and $\{p_{\text{end}}^2\}$ split $X$. This means that after round~$i_X=k+8$ is done, $X$ is discrete. As each of those $\mathcal{O}(k)$ round can have costs at most $\mathcal{O}(2^k \cdot k)$ the total costs until $X$ is discrete are in $\mathcal{O}(2^k \cdot k^2)$.

\begin{lemma}
	\label{lemma:Queue>Stack_runtime_part1}
	
	When applying a queue based version of Algorithm~\ref{alg:refine} with preservation or with removal on $G^{\text{Q}}_k$, the following holds:
	\begin{itemize}
		\item In the first $k+8$ rounds, $X$ and $Y$ will never be split.
		\item After the $k+8$-th round, $X$ is discrete.
		\item Executing the first $k+9$ rounds takes time $\mathcal{O}(2^k \cdot k^2)$
	\end{itemize}
	
\end{lemma}

\noindent
The next round splits $\mathcal{X}$ into the blocks of level~$k$, the one thereafter does the same with $\mathcal{Y}$ and after round~$k+11$ all four layers $X$, $\mathcal{X}$, $\mathcal{Y}$ and $Y$ are split into the blocks level~$k$ equaling their partition of $\alpha_\infty$. Each of these three rounds looks at $\mathcal{O}(2^k \cdot k^2)$ edges, so the total running time from the previous lemma is not changed asymptotically.

\begin{lemma}
	\label{lemma:Queue>Stack_runtime_part1&2}
	
	When applying a queue based version of Algorithm~\ref{alg:refine} on $G^{\text{Q}}_k$, the following holds:
	\begin{itemize}
		\item For coloring $\alpha'$ after the $k+11$-th round it holds that $X_q^k,Y_q^k,\mathcal{X}_q^k,\mathcal{Y}_q^k \in \pi_{\alpha'}$
		\item Executing the first $k+12$ rounds takes time $\mathcal{O}(2^k \cdot k^2)$
	\end{itemize}
	
\end{lemma}

\noindent
Finishing the color refinement does also take at most $\mathcal{O}(2^k \cdot k^2)$ time because the subclasses of $\mathcal{X}$ and $\mathcal{Y}$ are no longer handled and therefore the subgraph considered by the algorithm has a size of $\mathcal{O}(2^k \cdot k)$, meaning that $\mathcal{O}(2^k \cdot k^2)$ is an upper bound for the computation of the coarsest stable coloring.

\begin{lemma}
	\label{lemma:Queue>Stack_runtime_part3}
	
	Let $\alpha'$ be a coloring for $G^{\text{Q}}_k$ with $X_q^k,Y_q^k,\mathcal{X}_q^k,\mathcal{Y}_q^k \in \pi_{\alpha'}$. Computing the coarsest stable coloring refining $\alpha'$ with a queue based color refinement takes time $\mathcal{O}(2^k \cdot k^2)$.
	
\end{lemma}

\noindent
With these three parts of the algorithm execution taking each time $\mathcal{O}(2^k \cdot k^2)$, the total running time is also $\mathcal{O}(2^k \cdot k^2)$.

\begin{theorem}
	\label{lemma:Queue>Stack_runtime_queue}
	
	Algorithm~\ref{alg:refine}, using a queue as worklist, takes time $\mathcal{O}(2^k \cdot k^2)$ on $G^{\text{Q}}_k$.
	
\end{theorem}

\begin{proof}
	
	Follows directly from Lemmas~\ref{lemma:Queue>Stack_runtime_part1}~,~\ref{lemma:Queue>Stack_runtime_part1&2}~and~\ref{lemma:Queue>Stack_runtime_part3}.
	
\end{proof}

\section{Priority Queues}
\label{sec:priorityQueues}

For the sake of completeness, we briefly discuss priority queues as worklists.
In particular, we consider a worklist which always chooses the smallest or largest class w.r.t.\ which the algorithm has not refined yet.
We include this since the use (some form) of a priority is employed by the 
state-of-the-art tool \emph{Traces}~\cite{practicaliso2}. In fact that tool combines a priority queue with a stack based strategy.
Specifically, a constant amount of classes are taken from the top of the stack and from those the smallest one is chosen to continue with.
If there is no unique smallest class, the one which appears first on the stack is chosen.

For both options regarding priority queues, we construct a graph class, on which the refinement has costs of $\Omega(m \log(n))$, whereas a linear color refinement is possible with another strategy for the worklist.
Both graph classes are based on the simplified graphs from Section~\ref{sec:Stack>Queue}.
A priority queue that always gives the largest or smallest class can be easily forced into the slow behavior by changing the sizes of the inner layers $X,\mathcal{X},\mathcal{Y},Y$, such that on each level all of $\mathcal{X}$ is handled before continuing with the next level.

For a maximum priority queue, it is enough to duplicate each vertex in $X$.
Then, after refining w.r.t.\ a level~$l$ block of $X$, the refinements w.r.t.\ the corresponding blocks of $\mathcal{X}$ and $\mathcal{Y}$ are done, but the remaining blocks of $X$ are preferred over the new $Y$ block, as they have double the size.
As this happens for all blocks of $X$, $\Theta(2^l)$ blocks of $\mathcal{X}$ are handled before any block of $Y$, and we get worst case behavior.

For minimum priority queues, we conversely change the set $\mathcal{Y}$ such that each block of it is larger than the corresponding block of $\mathcal{X}$.
Since in this case, all blocks of $X$ are immediately handled for each level, the worklist can now choose from the $\mathcal{X}$ blocks.
Each $X$ block distinguishes the corresponding $\mathcal{Y}$ block, but the new $\mathcal{Y}$ blocks are never chosen before all blocks of $\mathcal{X}$ have been handled. Therefore, we again get worst case costs.

\end{document}